\documentclass[12pt]{article}
\usepackage{amsmath}
\usepackage{amsthm}
\usepackage{amsfonts}
\usepackage{color}
\usepackage{graphicx}
\usepackage{float}
\usepackage{amsmath}
\usepackage{amsfonts}
\usepackage{multirow}
\usepackage{setspace}
\usepackage{array}
\usepackage{amssymb}
\usepackage{lscape}
\usepackage[authoryear]{natbib}
\usepackage{color}
\usepackage[usenames,dvipsnames]{xcolor}
\usepackage[colorlinks]{hyperref}
\usepackage{rotating}
\usepackage{tikz}
\usepackage{framed}
\usepackage{cleveref}
\usepackage{autonum}
\usepackage{ulem}

\usepackage{caption}
\usepackage{subcaption}
\usepackage{stackengine}

\usepackage{bbm}
\usepackage{xfrac}

\usepackage{microtype}

\usetikzlibrary{backgrounds}
\usetikzlibrary{arrows,shapes}
\usetikzlibrary{tikzmark}
\usetikzlibrary{calc}
\usetikzlibrary{arrows.meta}
\usetikzlibrary{positioning}

\usepackage{tcolorbox}

\usepackage[margin=1in]{geometry}
\usepackage{caption}

\usepackage{attrib} 

\newtheorem{theorem}{{Theorem}}[section]
\newtheorem{proposition}[theorem]{Proposition}
\newtheorem{lemma}[theorem]{Lemma}
\newtheorem{definition}[theorem]{Definition}

\newtheorem{example}[theorem]{{Example}}
\newtheorem{corollary}[theorem]{{Corollary}}

\hypersetup{
    unicode=false,          
    pdftoolbar=true,        
    pdfmenubar=true,        
    pdffitwindow=false,     
    pdfstartview={FitH},    
    pdftitle={Sequential Cursed Equilibrium},    
    pdfauthor={Author},     
    pdfsubject={Subject},   
    pdfcreator={Creator},   
    pdfproducer={Producer}, 
    pdfkeywords={keyword1, key2, key3}, 
    pdfnewwindow=true,      
    colorlinks=true,       
    linkcolor=blue,          
    citecolor=blue,        
    filecolor=magenta,      
    urlcolor=cyan           
}

\DeclareMathOperator*{\argmax}{\arg\!\max}
\DeclareMathOperator{\E}{E}

\newcommand{\nature}{\lambda}
\newcommand{\measure}{\mu}
\renewcommand{\epsilon}{\varepsilon}

\newif\ifcomments
\commentstrue   

\usepackage{subfiles}

\title{Sequential Cursed Equilibrium\thanks{We thank Vince Crawford, Ignacio Esponda, Xavier Gabaix, Yannai Gonczarowski, Philippe Jehiel, Scott Kominers, David Laibson, Muriel Niederle, Matthew Rabin, and Ran Spiegler for helpful comments. We thank Meng-Jhang Fong, Po-Hsuan Lin, and Thomas Palfrey for correspondence that helped clarify the difference between their work and ours. We thank Arjun Nageswaran for research assistance. All errors remain our own.}}
\author{Shani Cohen\thanks{Harvard University. Email: \url{shani_cohen@g.harvard.edu}} \hspace{10mm} Shengwu Li\thanks{Harvard University. Email: \url{shengwu_li@fas.harvard.edu}}}

\date{First posted: December 13, 2022 \\ Current version: \today}

\begin{document}

\maketitle

\begin{abstract}
    We propose an extensive-form solution concept, with players that neglect information from hypothetical events, but make inferences from observed events. Our concept modifies cursed equilibrium \citep{eyster2005cursed}, and allows that players can be cursed about endogenous information.
\end{abstract}

\section{Introduction}\label{sec:introduction}

Standard game theory posits that players look before they leap. When voting, they condition on their vote being pivotal. When placing a limit order for an asset, they condition on the limit price being met. When bidding in an auction, they condition on their bid being the highest. Generally, if it would be informative to observe an event, then players account for that information when thinking hypothetically about that event. \cite{savage1972foundations} writes,
\begin{quotation}
    In view of the ``look before you leap" principle, acts and decisions, like events, are timeless. The person decides ``now" once for all; there is nothing for him to wait for, because his one decision provides for all contingencies.
\end{quotation}
However, Savage concedes that, carried to its logical extreme, ``the task implied by making such a decision is not even remotely resembled by human possibility."

In practice, people often do not look before they leap. Instead, they neglect the implications of hypothetical events, but make inferences based on observed events---they fail to condition on the contingency that their vote is pivotal, that their limit order is filled, or that their bid is the highest, but they respond appropriately if the relevant uncertainty is resolved \citep{esponda2014hypothetical, martinez2019failures, moser2019hypothetical, ngangoue2021learning,niederle2023}. Similarly, lab experiments find that, under common values, dynamic auctions lead to less cursed behavior than sealed-bid auctions, because observing the behavior of other bidders helps subjects to reason about what the other bidders know \citep{levin1996revenue,levin2016separating}.

We propose a solution concept with players who look after they leap. That is, they learn from observing other players' actions, but overlook this information when thinking hypothetically.

Our theory builds on cursed equilibrium, a canonical solution concept for Bayesian normal-form games \citep{eyster2005cursed}. Under cursed equilibrium, players neglect the link between what other players know and what other players do. Suppose that we have a Bayesian game; nature draws a payoff-relevant state, and each player observes a partition of the state space. Fix a strategy profile. For each type of each player, there is a conditional joint distribution on opponent actions and the state. Suppose that each player understands the marginal distribution of opponent actions and the marginal distribution of the state, conditional on their own type, but responds as if these distributions are independent of each other. A strategy profile is a \textbf{cursed equilibrium (CE)} if it is a fixed point of the resulting best-response correspondence.

To fix ideas, consider the following trading game from \cite{eyster2005cursed}.  Two players decide simultaneously whether to trade; each can accept or decline. Trade occurs if and only if both accept. There are three states of the world $\{\omega_1, \omega_2, \omega_3\}$, each with prior probability $\frac{1}{3}$. \Cref{tab:trading_payoffs} specifies the payoffs and each player's information.
\begin{table}[h]
\captionsetup{width=.75\textwidth}
\caption{Payoffs and types for the trading game.}\label{tab:trading_payoffs}
\begin{center}
\begin{tabular}{l| ccc}
         & $\omega_1$ & $\omega_2$            & $\omega_3$ \\
\hline
trade    & $3, -3$    & $1, -1$ & $-3, 3$    \\
no trade & $0,0$      & $0,0$                 & $0,0$     \\
\hline
player $1$ types  &            \multicolumn{2}{c|}{$t_1$} & \multicolumn{1}{c|}{$t'_1$} \\
\cline{2-4}
player $2$ types  &            \multicolumn{1}{c|}{$t_2$} & \multicolumn{2}{c|}{$t'_2$} \\
\cline{2-4}
\end{tabular}
\end{center}
\end{table}

This trading game is dominance-solvable. Player $1$ wants to trade in states $\omega_1$ and $\omega_2$, but not in $\omega_3$. Thus, it is a dominant strategy for type $t_1$ to accept and for type $t'_1$ to decline. Player $2$ only wants to trade in state $\omega_3$, so it is a dominant strategy for type $t_2$ to decline. Now the last step requires type $t'_2$ to look before they leap: Conditional on player $1$ accepting, the state is in $\{\omega_1,\omega_2\}$ for sure, so type $t'_2$ should decline. In every Bayesian Nash equilibrium (BNE) of this game, there is no trade.

By contrast, trade can arise in cursed equilibrium. In one CE, type $t_1$ accepts, type $t'_1$ declines, and type $t_2$ declines, just as before. Suppose player $2$'s type is $t'_2$. Conditional on $t'_2$, player $1$ accepts with probability $\frac{1}{2}$ and declines with probability $\frac{1}{2}$, and the state is $\omega_2$ with probability $\frac{1}{2}$ and $\omega_3$ with probability $\frac{1}{2}$. Treating these as independent, player $2$'s payoff from accepting is $-1$ with probability $\frac{1}{4}$, $3$ with probability $\frac{1}{4}$, and $0$ with probability $\frac{1}{2}$. Thus, type $t'_2$ accepts, resulting in trade in state $\omega_2$. Moreover, this is the unique CE in undominated strategies. Lab experiments find high rates of trade in similar dominance-solvable games \citep{sonsino2002likelihood,sovik2009strength,rogers2009heterogeneous,brocas2014imperfect}.

However, cursed equilibrium players do not learn from observed events---they do not look even after they leap. 
To illustrate, we modify the trading game so that the players decide in sequence. Player $1$ moves first. If player $1$ declines, then the game ends, and if player $1$ accepts, then player $2$ observes $1$'s action and chooses to accept or decline. \Cref{fig:trading_seq} depicts the extensive game.

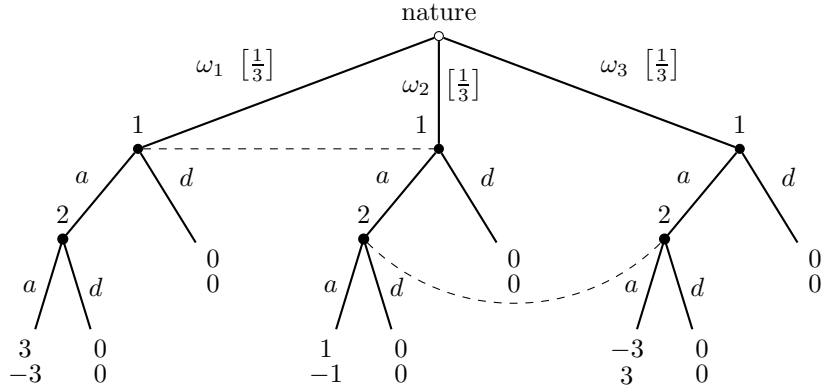
\begin{figure}[h]
    \centering
\begin{tikzpicture}[font=\footnotesize,edge from parent/.style={draw,thick}]
    \tikzstyle{solid}=[circle,draw,inner sep=1.2,fill=black];
    \tikzstyle{hollow}=[circle,draw,inner sep=1.2]
    \tikzstyle{level 1}=[level distance=15mm,sibling distance=40mm]
    \tikzstyle{level 2}=[level distance=12mm,sibling distance=20mm]
    \tikzstyle{level 3}=[level distance=12mm,sibling distance=10mm]
     \node(0)[hollow]{}
     child{node[solid]{}
       child{node[solid]{}
        child{node[below]{\stackon{$-3$}{$3$}} edge from parent node[left]{$a$}}
        child{node[below]{\stackon{$0$}{$0$}} edge from parent node[right]{$d$}}
        edge from parent node[above left]{$a$}
        }
       child{node[below]{\stackon{$0$}{$0$}} 
        edge from parent node[above right]{$d$}
        }
    edge from parent node[above left]{$\omega_1 \ \left[\frac{1}{3}\right]$}
     }
     child{node[solid]{}
       child{node[solid]{}
        child{node[below]{\stackon{$-1$}{$1$}} edge from parent node[left]{$a$}}
        child{node[below]{\stackon{$0$}{$0$}} edge from parent node[right]{$d$}}
        edge from parent node[above left]{$a$}
        }
       child{node[below]{\stackon{$0$}{$0$}} 
        edge from parent node[above right]{$d$}
        }
    edge from parent node[xshift=.5mm,yshift=1mm]{$\omega_2 \ \left[\frac{1}{3}\right]$}
     }
     child{node[solid]{}
       child{node[solid]{}
        child{node[below]{\stackon{$3$}{$-3$}} edge from parent node[left]{$a$}}
        child{node[below]{\stackon{$0$}{$0$}} edge from parent node[right]{$d$}}
        edge from parent node[above left]{$a$}
        }
       child{node[below]{\stackon{$0$}{$0$}} 
        edge from parent node[above right]{$d$}
        } 
    edge from parent node[above right]{$\omega_3 \ \left[\frac{1}{3}\right]$}
     }
     ;
     \draw[dashed](0-1)to(0-2);
     \draw[dashed,bend right=45](0-2-1)to(0-3-1);
     \node[above,yshift=2]at(0){nature}; \node[above,yshift=2]at(0-1){1};
     \node[above left,yshift=2]at(0-2){1};
     \node[above,yshift=2]at(0-3){1};
     \foreach \x in {1,2,3} \node[above,yshift=2]at(0-\x-1){2};
\end{tikzpicture}
    \caption{The sequential trading game}
    \label{fig:trading_seq}
\end{figure}

Suppose that player $1$ plays their dominant strategy, accepting at type $t_1 = \{\omega_1,\omega_2\}$ and declining at type $t_2 = \{\omega_3\}$. Suppose that player $2$ has type $t'_2 = \{\omega_2, \omega_3\}$ and sees that player $1$ accepts. Observing player $1$'s action reduces the uncertainty for player 2. Player $2$ can condition on that action, and it is intuitive that this helps them to realize that the state is not $\omega_3$---after all, if the state was $\omega_3$, then player $1$ would have declined for sure. Plausibly, therefore, player $2$ will infer that the state is $\omega_2$, and decline to trade. However, this is not the cursed equilibrium prediction, because the sequential trading game has the same Bayesian normal form as the simultaneous game, and hence the same CE with trade in state $\omega_2$.

\cite{eyster2005cursed} argue that the cursed equilibrium prediction is implausible in such cases, writing:
\begin{quotation}
[\ldots] players are probably more likely to ignore the informational content of other players’ actions when they have not actually observed these actions than when they have; observing actions seems likely to induce more strategic sophistication. Hence, players in certain sequential games may be less cursed than they would be in corresponding simultaneous-move games.
\end{quotation}

Another limitation of cursed equilibrium is that players are cursed about types, but not about endogenous information. Suppose that we modify the simultaneous trading game by replacing nature with player $0$, who plays first with action set $A_0 = \{{\omega}_1, {\omega}_2, {\omega}_3\}$. Player $0$'s action determines the payoffs from trade. Player $0$ has constant payoffs, and is thus willing to play each action with probability $\frac{1}{3}$. Players $1$ and $2$ learn about player $0$'s action, with the same partition as in \Cref{tab:trading_payoffs}. Intuitively, replacing nature with a fictitious player should not change behavior---but under cursed equilibrium, it does. In the Bayesian normal form of the game with the fictitious player, every player's type set is singleton, so CE coincides with BNE.\footnote{In the Bayesian normal form, player $1$ has four actions $A_1 = \left\{a_{\omega_1,\omega_2}, d_{\omega_1,\omega_2}\right\} \times \left\{a_{\omega_3}, d_{\omega_3}\right\}$ and similarly $A_2 = \left\{a_{\omega_1}, d_{\omega_1}\right\} \times \left\{a_{\omega_2,\omega_3}, d_{\omega_2,\omega_3}\right\}$.} Thus, in every CE there is no trade. \cite{eyster2005cursed} point out this limitation, writing:
\begin{quotation}
Treating “exogenous” and “endogenous” private information differently not only seems to us intuitively and psychologically wrong, but also creates some highly artificial differences in predictions based on the way a game is formally written down.
\end{quotation}

We adapt cursed equilibrium to extensive games, so that players neglect the implications of hypothetical events, but make inferences based on observed events. In essence, we modify cursed equilibrium as follows:
\begin{quotation}
    At each \sout{type} \textit{information set}, the player understands the marginal distributions of the other players' actions and of \sout{the state} \textit{nature moves}, conditional on that \sout{type} \textit{information set}, but responds as if the other players’ actions are independent of their \sout{types} \textit{information sets}.
\end{quotation}
This approach draws no distinction between exogenous and endogenous information.

What does it mean for actions to be `independent of information sets'? Given a game tree, only some partitions of non-terminal histories can be information sets for an extensive game of perfect recall. There exists a coarsest such partition, that we denote $\mathcal{F}$. A set of histories is a \textbf{coarse set} if it is a cell of $\mathcal{F}$. For instance, in \Cref{fig:trading_seq}, the union of player $1$'s information sets is a coarse set and the union of player $2$'s information sets is a coarse set. (Coarse sets containing histories assigned to nature are singleton.) By construction, all histories in a coarse set have the same player and the same available actions. A strategy is \textbf{coarse} if it plays the same way at all histories in the same coarse set. We require that each player responds as if their opponents are playing coarse strategies.

What does it mean to `condition on' an information set? Fix some strategy profile $\sigma_N$, where $N$ denotes the set of players. At information set $I$, the active player forms a conjecture about another player's (or nature's) behavior at coarse set $F \in \mathcal{F}$. We require that the conjecture at $I$ about $F$ matches the empirical frequency of actions at $F$ under $\sigma_N$, conditional on the event that $I$ is reached and $F$ is reached. We call the tuple consisting of a strategy and a system of conjectures an \textbf{assessment}, and say that an assessment is \textbf{cursed-plausible} if this requirement holds `whenever possible', that is, whenever the relevant event has positive probability.

Cursed equilibrium takes the marginal distribution over the state and over opponent actions, conditional on the type. By contrast, cursed plausibility updates these marginal distributions to match the empirical frequencies at each information set. Thus, when one player observes another player's actions, they infer the information signalled by those actions.

To illustrate, suppose that player $1$ accepts at type $t_1 = \{\omega_1, \omega_2\}$ and declines at type $t'_1 = \{ \omega_3 \}$. In the simultaneous trading game (\Cref{tab:trading_payoffs}), the cursed-plausible conjecture at the information set corresponding to type $t'_2 = \{\omega_2,\omega_3\}$ is that nature plays $\left(\omega_2 \ \frac{1}{2},\omega_3 \ \frac{1}{2}\right)$ and that player $1$ plays the coarse strategy $\left(a \ \frac{1}{2},d \ \frac{1}{2}\right)$. By contrast, in the sequential trading game (\Cref{fig:trading_seq}), the cursed-plausible conjecture at the information set $\{\omega_2  a, \omega_3  a \}$ is that nature plays $\omega_2$ for sure and that player $1$ plays $a$ for sure. In this way, cursed plausibility implies that players neglect information from unobserved (or future) events, but make rational inferences from observed events. 

However, cursed-plausibility sometimes does not pin down all the conjectures. Returning to \Cref{fig:trading_seq}, suppose instead that player $1$ plays $d$ at information set $I = \{\omega_1,\omega_2\}$ and player $2$ plays $d$ at $\{\omega_1 a\}$ and $a$ at $\{\omega_2 a, \omega_3 a\}$. Then the event that $I$ is reached and player $2$'s coarse set $F=\{\omega_1 a, \omega_2 a, \omega_3 a\}$ is reached has zero probability. Thus, cursed-plausibility does not restrict what player $1$ believes at $I$ about the probability that $2$ plays $a$ at $F$. The problem is that player $1$'s conjecture about player $2$'s behavior should be tied to actual behavior.

We can use trembles to pin down conditional probabilities, as in sequential equilibrium \citep{kreps1982sequential}. Specifically, an assessment is \textbf{cursed-consistent} with strategy profile $\sigma_N$ if it is the limit of the cursed-plausible assessments induced by a sequence of fully mixed strategy profiles $(\sigma_N^k)_{k = 1}^\infty$ that converges to $\sigma_N$. Cursed-consistency implies cursed-plausibility.\footnote{In the preceding example, at information set $I = \{\omega_1,\omega_2\}$, any cursed-consistent conjecture for player $1$ specifies that $2$ plays $a$ with probability $\frac{1}{2}$ at coarse set $F=\{\omega_1 a, \omega_2 a, \omega_3 a\}$.}

An assessment is \textbf{locally rational} if for each player $n$ and each information set $I$ at which $n$ is called to play, the player's strategy and self-conjecture coincide on $I$, and the player's self-conjecture maximizes their continuation payoff given the conjectures about others.

A \textbf{sequential cursed equilibrium (SCE)} is an assessment that is locally rational and cursed-consistent.

For the sake of comparison, let us define \textbf{independently cursed equilibrium (ICE)} to be as in CE, but with players responding as if the opponents' actions are independent of \textit{each other}.\footnote{Recall that CE takes the marginal over opponent action profiles, so players understand the correlations between different opponents' actions.} ICE and SCE are equivalent on simultaneous Bayesian games. As a corollary, CE and SCE are equivalent on two-player simultaneous Bayesian games. Thus, the simultaneous trading game has a unique SCE in undominated strategies, with trade in state $\omega_2$.

SCE players neglect the implications of hypothetical events, but make inferences from observed events. For instance, SCE predicts no trade in the sequential trading game (\Cref{fig:trading_seq}). If player $1$ accepts at $\{\omega_1, \omega_2\}$ and declines at $\{\omega_3\}$, then at player $2$'s information set $\{\omega_2 a, \omega_3 a\}$, cursed-consistency implies that player $2$ believes that nature has played $\omega_2$ for sure. Thus, upon observing that player $1$ accepts, player $2$ declines to trade.

By construction, SCE treats exogenous and endogenous information the same way; strategies are coarse over information sets, not just over types. For instance, the SCE prediction for the trading game is unchanged when nature is replaced by a fictitious player playing according to the same distribution.\footnote{The simultaneous trading game with a fictitious player has a SCE in which the player $0$ plays each action with probability $1/3$, and players $1$ and $2$ trade as before.} SCE permits cursedness about new kinds of information, such as information that one player has about other players, and information about the state that players chose to acquire.\footnote{For instance, SCE allows for cursedness in auctions with costly entry. In such models, bidders choose whether to enter, and the bidders who enter get estimates of the object's value \citep{levin1994equilibrium}.}

We compare the predictions of BNE, CE, and SCE in \Cref{tab:example_comparison}, rows $1$ to $3$.

\begin{table}[h]
\captionsetup{width=.75\textwidth}
\caption{Predictions for the trading games. $\checkmark$ indicates that there exists an equilibrium with a positive probability of trade. $\lozenge$ indicates that the solution concept is not defined for that game.}\label{tab:example_comparison}
\begin{center}
\begin{tabular}{l|ccc}
    & Simultaneous & Sequential & Fictitious player \\
    \hline
Bayesian Nash equilibrium &              &                &            \\
Cursed equilibrium &       $\checkmark$       &     $\checkmark$    &            \\
Sequential cursed equilibrium &      $\checkmark$        & &       $\checkmark$    \\
ABEE with analogy partition $\mathcal{F}$ &      $\checkmark$        & $\checkmark$ &       $\checkmark$    \\
$\chi$-CSE with $\chi \geq 1/2 $ &      $\checkmark$        &  $\checkmark$  &       $\lozenge$   \\
$\chi$-CSE with $\chi < 1/2 $ &             &      &       $\lozenge$   
\end{tabular}
\end{center}
\end{table}

SCE incorporates ideas from analogy-based expectation equilibrium (ABEE) \citep{jehiel2005analogy}. ABEE posits that players bundle histories into analogy classes, and best-respond to the weighted average of behavior in each analogy class. In principle, the novel partition $\mathcal{F}$ could define analogy classes for ABEE.\footnote{\cite{jehiel2005analogy}, \cite{jehiel2008revisiting}, and \cite{jehiel2022analogy} suggest various ways to pin down the analogy classes for ABEE, none of which correspond to the partition $\mathcal{F}$ proposed here.} However, under ABEE, the weighted average is proportional to the \textit{ex ante} probability, under the true strategy profile, that each history in the analogy class is reached. By contrast, SCE posits that players best-respond to the \textit{interim} distribution of behavior conditional on each information set. Thus, when SCE players receive new information about play, they change their beliefs to reflect the conditional averages. This feature is crucial for distinguishing hypothetical and observed events. For instance, ABEE predicts trade in all three versions of the trading game, as \Cref{tab:example_comparison} reports.

In contemporaneous work, \cite{FPP2023} proposed a different generalization of cursed equilibrium to dynamic games. Their solution concept, $\chi$-cursed sequential equilibrium ($\chi$-CSE), is defined for multi-stage games with observed actions.\footnote{This class rules out the trading game with the fictitious player, as well as auctions with entry.} $\chi$-CSE starts from the following premise: Cursed players in a Bayesian game neglect how other players' actions depend on other players' types. We can generalize this to dynamic games by having cursed players be slow to update about opponent types, after observing the public history.

Formally, under $\chi$-CSE, players respond as though the other players' actions are independent of the other players' types with probability $\chi \in [0,1]$, independently across stages. Thus, players treat the public history as less informative than it truly is.

$\chi$-CSE posits that players learn fully from exogenous information (their own type), but under-infer from endogenous information (the public history). By contrast, SCE posits that players learn from exogenous and endogenous information in the same way.

In some games, $\chi$-CSE and SCE make starkly different predictions. Under $\chi$-CSE, players do not learn differently from hypothetical versus observed events. For example, $\chi$-CSE predicts that the simultaneous and sequential trading games result in the same behavior, as \Cref{tab:example_comparison} reports. $\chi$-CSE is not a refinement of SCE, nor \textit{vice versa}. For instance, $\chi$-CSE diverges from SCE on the games studied in \Cref{sec:pivotal_voting} and \Cref{sec:learning_from_prices}.\footnote{$\chi$-CSE is not defined for the game in \Cref{sec:auctions_with_learning}.} It predicts that the experimental manipulations do not affect behavior. We contrast SCE and $\chi$-CSE in \Cref{app:SCE_CSE_diff}.\footnote{\cite{fong2023note} also compare $\chi$-CSE and SCE in detail.}


The paper proceeds as follows: We formally define sequential cursed equilibrium and state an existence result (\Cref{sec:def_existence}). We show that SCE reduces to ICE in simultaneous Bayesian games (\Cref{sec:simultaneous_bayesian}). We generalize SCE to allow partial cursedness (\Cref{sec:chi_psi_SCE}). We then study three experiments about hypothetical thinking that falsify both BNE and CE, and find that the treatment effects are consistent with SCE (\Cref{sec:applications}). We adapt the main ideas to dynamic auctions (\Cref{sec:dynamic_auctions}). We discuss limitations in \Cref{sec:discussion}. Proofs omitted from the main text are in \Cref{app:proofs}.

\section{Definitions and general results}

\subsection{Definition and existence of sequential cursed equilibrium}\label{sec:def_existence}

We use the standard formulation for extensive games \cite[p.\ 226-227]{mas1995microeconomic}. Throughout, we restrict attention to games of perfect recall. The set of players is denoted $N$ and the set of histories is denoted $H$, a tree ordered by binary precedence relation $\prec$. We write $h \preceq h'$ if $h \prec h'$ or $h$ and $h'$ are identical. The collection of information sets is denoted $\mathcal{I}$; these are a partition of the non-terminal histories. $\mathcal{I}_n$ denotes the collection of information sets assigned to player $n$. We represent objective uncertainty using a nature player $\nature$. As is standard, we require that nature's information sets, $I \in \mathcal{I}_\nature$, are singleton. \Cref{tab:notation} introduces and summarizes notation.

\begin{table}[h]
\caption{Notation for extensive games}
\label{tab:notation}
\begin{center}
\begin{tabular}{lcc}
\hline
\hline
Name                                    & Notation                                  & Representative Element \\
\hline
Set of players (excluding nature)       & $N$                                       & $n$                    \\
Nature                                  & $\nature$                                       &                        \\
Set of histories                        & $H$                                       & $h$                    \\
Order on histories                       & $\prec$                                       &                    \\
Set of terminal histories                        & $Z$                                       & $z$                    \\
Collection of information sets          & $\mathcal{I}$                             & $I$                    \\
Active player function                  & $P: \mathcal{I} \rightarrow N \cup \{\nature\}$ &                        \\
Information sets assigned to player $n$ & $\mathcal{I}_n$                           & $I$                    \\
Histories assigned to player $n$ & $H_n$                           & $h$                    \\
Set of actions                          & $A$                                       & $a$                    \\
Actions available at $I$                & $A(I)$                                    & $a$        \\       
\hline
\end{tabular}
\end{center}
\end{table}

A \textbf{strategy} $\sigma_n$ is a function from the collection of information sets $\mathcal{I}_n$ to distributions over actions, satisfying $\sigma_n(I) \in \Delta (A(I))$. We use $\sigma_n(a,I)$ to denote the probability that action $a$ is played at $I$ under $\sigma_n$. Nature, denoted $\lambda$, plays according to some fixed fully mixed strategy $\sigma_\nature$. A strategy profile for $N \cup \{\nature\}$ (the set of players and nature) is denoted $\sigma$, a strategy profile for the set of players is denoted $\sigma_N$, and a strategy profile for $(N \setminus \{n\}) \cup \{\nature\}$ is denoted $\sigma_{-n}$. The expected utility to player $n$ from strategy profile $\sigma$ is denoted $u_n(\sigma)$. The expected utility to player $n$ when we start play at history $h$ and play continues according to $\sigma$ is denoted $u_n(h,\sigma)$.

We extend the relation $\prec$ to arbitrary sets $Q, Q' \subseteq H$, writing $Q \prec Q'$ if there exists $h \in Q$ and $h' \in Q'$ such that $h \prec h'$. Similarly, we extend the relation $\preceq$, writing $Q \preceq Q'$ if there exists $h \in Q$ and $h' \in Q'$ such that $h \preceq h'$.

We say that information sets $I$ and $I'$ are \textbf{compatible} if $I \preceq I'$ or $I' \preceq I$. \Cref{fig:running_1} depicts a game that we use as a running example. Player $1$'s information set and player $2$'s leftmost information set are not compatible. We use $\sigma_n^I$ to denote a \textbf{partial strategy for player $n$ at $I$}; this is a strategy with its domain restricted to information sets compatible with $I$, and is defined for any player $n$ (not only the active player at $I$).

\begin{figure}[h]
\centering
\begin{subfigure}[T]{.48\textwidth}
  \centering
\begin{tikzpicture}[font=\footnotesize,edge from parent/.style={draw,thick}]
    \tikzstyle{solid}=[circle,draw,inner sep=1.2,fill=black];
    \tikzstyle{hollow}=[circle,draw,inner sep=1.2]
    \tikzstyle{solid_red}=[circle,draw=red,inner sep=1.7,fill=red];
    \tikzstyle{solid_mag}=[circle,draw=magenta,inner sep=1.7,fill=red];
    \tikzstyle{level 1}=[level distance=18mm,sibling distance=24mm]
    \tikzstyle{level 2}=[level distance=12mm,sibling distance=12mm]
    \tikzstyle{level 3}=[level distance=12mm,sibling distance=12mm]
     \node(0)[hollow]{}
     child[level distance=30mm,sibling distance=20mm]{node[solid]{}
       child{ node[below]{}
        edge from parent node[left]{$l$}
        }
       child{node[below]{} 
        edge from parent node[right]{$r$}
        }
     }
     child{node[solid]{}
       child{node[below]{} 
        edge from parent node[left]{$x$}
        }
        child{node[solid]{}
        child{node[below]{} edge from parent node[left]{$l$}}
        child{node[below]{} edge from parent node[right]{$r$}}
        edge from parent node[right]{$y$}
        }
    edge from parent node[xshift=.5mm, name=omega2]{$\omega_2 \ {\color{blue} [.2]}$}
     }
     child{node[solid]{}
       child{node[below]{} 
        edge from parent node[left]{$x$}
        }
       child{node[solid]{} 
        child{node[below]{} edge from parent node[left]{$l$}}
        child{node[below]{} edge from parent node[right]{$r$}}
        edge from parent node[right]{$y$}
        }
    edge from parent node[right, xshift=2mm] 
    {$\omega_3{\color{blue} [.4]}$}
     }
     ;
     \node[left=.5mm of omega2] {$\omega_1 {\color{blue} [.4]}$};
      \draw[dashed](0-2)to(0-3);
     \node[above,yshift=2]at(0){nature}; \node[above,yshift=2]at(0-1){2};
     \node[above right,yshift=2]at(0-2){1};
     \node[above,yshift=2]at(0-3){1};
     \foreach \x in {2,3} \node[above right,yshift=2]at(0-\x-2){2};
    \node[below]at(0-1-1){\stackon{$1$}{$0$}};
    \node[below]at(0-1-2){\stackon{$0$}{$0$}};
    \node[below]at(0-2-1){\stackon{$0$}{$1$}};
    \node[below]at(0-3-1){\stackon{$0$}{$1$}};
    \node[below]at(0-2-2-1){\stackon{$1$}{$0$}};
    \node[below]at(0-2-2-2){\stackon{$0$}{$6$}};
    \node[below]at(0-3-2-1){\stackon{$0$}{$0$}};
    \node[below]at(0-3-2-2){\stackon{$1$}{$0$}};
    \draw[blue, ->, line width=1.5pt, -{Triangle[length=8pt]},] (0-1) -- ($(0-1)!.97!(0-1-1)$);
    \draw[blue, ->, line width=1.5pt, -{Triangle[length=8pt]},] (0-2) -- ($(0-2)!.97!(0-2-2)$);
    \draw[blue, ->, line width=1.5pt, -{Triangle[length=8pt]},] (0-3) -- ($(0-3)!.97!(0-3-2)$);
    \draw[blue, ->, line width=1.5pt, -{Triangle[length=8pt]},] (0-2-2) -- ($(0-2-2)!.97!(0-2-2-1)$);
    \draw[blue, ->, line width=1.5pt, -{Triangle[length=8pt]},] (0-3-2) -- ($(0-3-2)!.97!(0-3-2-2)$);
\end{tikzpicture}
  \caption{Extensive game, with strategy profile in blue.}
  \label{fig:running_1}
\end{subfigure}
\begin{subfigure}[T]{.48\textwidth}
  \centering
\begin{tikzpicture}[font=\footnotesize,edge from parent/.style={draw,thick}]
    \tikzstyle{solid}=[circle,draw,inner sep=1.2,fill=black];
    \tikzstyle{hollow}=[circle,draw,inner sep=1.2]
    \tikzstyle{solid_red}=[circle,draw=red,inner sep=1.7,fill=red];
    \tikzstyle{solid_mag}=[circle,draw=magenta,inner sep=1.7,fill=red];
    \tikzstyle{level 1}=[level distance=18mm,sibling distance=24mm]
    \tikzstyle{level 2}=[level distance=12mm,sibling distance=12mm]
    \tikzstyle{level 3}=[level distance=12mm,sibling distance=12mm]
     \node(0)[hollow]{}
     child[level distance=30mm,sibling distance=20mm]{node[solid]{}
       child{ node[below]{}
        edge from parent node[left]{$l$}
        }
       child{node[below]{} 
        edge from parent node[right]{$r$}
        }
     }
     child{node[solid_red]{}
       child{node[below]{} 
        edge from parent node[left]{$x$}
        }
        child{node[solid]{}
        child{node[below]{} edge from parent node[left,xshift=.5mm]{$l{\color{blue} \left[\frac{1}{3}\right]}$}}
        child{node[below]{} edge from parent node[right,xshift=-.5mm]{$r{\color{blue} \left[\frac{2}{3}\right]}$}}
        edge from parent node[right]{$y$}
        }
    edge from parent node[xshift=.7mm, name=omega2]{$\omega_2 \ {\color{blue} \left[\frac{1}{3}\right]}$}
     }
     child{node[solid_red]{}
       child{node[below]{} 
        edge from parent node[left]{$x$}
        }
       child{node[solid]{} 
        child{node[below]{} edge from parent node[left,xshift=.5mm]{$l{\color{blue} \left[\frac{1}{3}\right]}$}}
        child{node[below]{} edge from parent node[right,xshift=-.5mm]{$r{\color{blue} \left[\frac{2}{3}\right]}$}}
        edge from parent node[right]{$y$}
        }
    edge from parent node[right, xshift=2mm]{$\omega_3 {\color{blue} \left[\frac{2}{3}\right]}$}
     }
     ;
     \node[left=.5mm of omega2] {$\omega_1 {\color{blue} [0]}$};
      \draw[dashed, red, line width=1.5](0-2)to(0-3);
     \node[above,yshift=2]at(0){nature}; \node[above,yshift=2]at(0-1){2};
     \node[above right,yshift=2]at(0-2){1};
     \node[above,yshift=2]at(0-3){1};
     \foreach \x in {2,3} \node[above right,yshift=2]at(0-\x-2){2};
     \node[below]at(0-1-1){\stackon{${\color{white}1}$}{${\color{white}0}$}};
     \node[below]at(0-1-2){\stackon{${\color{white}0}$}{${\color{white}0}$}};
     \node[below]at(0-2-1){\stackon{${\color{white}0}$}{${\color{white}1}$}};
     \node[below]at(0-3-1){\stackon{${\color{white}0}$}{${\color{white}1}$}};
     \node[below]at(0-2-2-1){\stackon{${\color{white}1}$}{${\color{white}0}$}};
     \node[below]at(0-2-2-2){\stackon{${\color{white}0}$}{${\color{white}6}$}};
     \node[below]at(0-3-2-1){\stackon{${\color{white}0}$}{${\color{white}0}$}};
     \node[below]at(0-3-2-2){\stackon{${\color{white}1}$}{${\color{white}0}$}};
    \draw[blue, ->, line width=1.5pt, -{Triangle[length=8pt]},] (0-2) -- ($(0-2)!.97!(0-2-2)$);
    \draw[blue, ->, line width=1.5pt, -{Triangle[length=8pt]},] (0-3) -- ($(0-3)!.97!(0-3-2)$);
\end{tikzpicture}
  \caption{Cursed-plausible conjecture at player $1$'s information set.}
  \label{fig:running_2}
\end{subfigure}%
\caption{A running example.}\label{fig:running_join_1_2}
\end{figure}
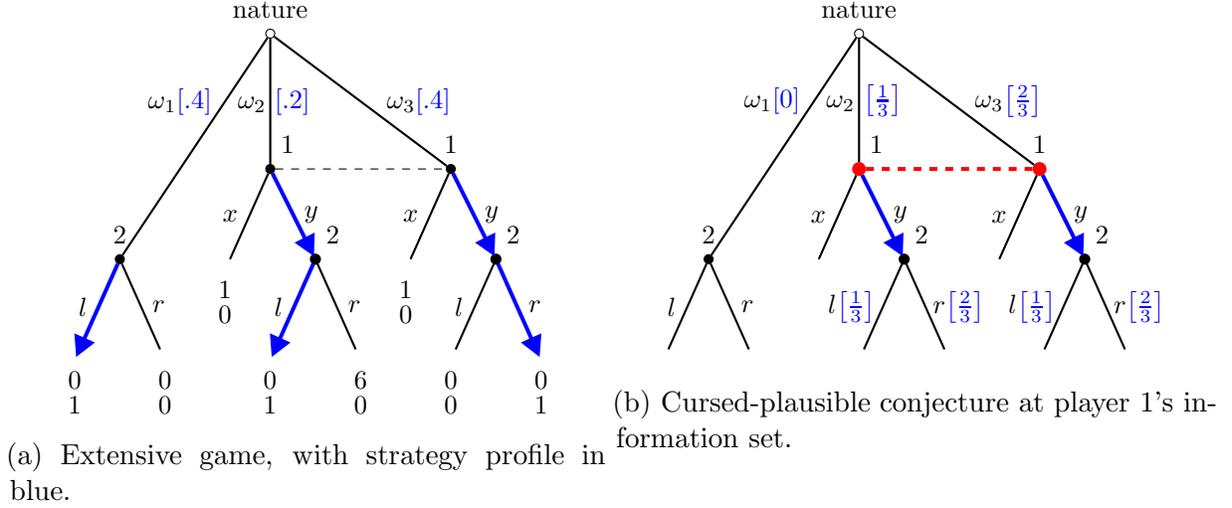

An extensive game is \textbf{well-behaved} if for all $h \in H$, the set $\{h' \in H: h' \prec h\}$ is finite. A partition of the non-terminal histories $H \setminus Z$ is \textbf{valid} if it defines information sets for some extensive game of perfect recall. If histories $h$ and $h'$ are in the same cell of a valid partition, then they have the same active player and the same available actions.

\begin{theorem}\label{thm:coarsest_valid}
For any well-behaved extensive game, there exists a valid partition $\mathcal{F}$ of $H \setminus Z$ that is coarser\footnote{Partition $\mathcal{J}$ is {coarser} than partition $\mathcal{J}'$ if for all $J'\in\mathcal{J}'$, there exists $J\in\mathcal{J}$ such that $J' \subseteq J$.} than any other valid partition.
\end{theorem}

The coarsest valid partition allows us to formalize the idea that each player acts independently of their own private information.\footnote{We prove \Cref{thm:coarsest_valid} by construction in \Cref{app:proof_coarsest_valid}. In contemporaneous work, \cite{nagel2024asif} provide an elegant proof by Zorn's lemma.} This is analogous to how, in cursed equilibrium, each player neglects how their opponent's behavior depends on their opponent's type.

A set of histories is a \textbf{coarse set} if it is a cell of $\mathcal{F}$, as defined implicitly in \Cref{thm:coarsest_valid}. We use $F(I)$ to denote $F \in \mathcal{F}$ such that $F \supseteq I$. We assume that coarse sets involving nature are singleton, \textit{i.e.} for all $I \in \mathcal{I}_\lambda$, we have $|F(I)| = 1$. For instance, the extensive game in \Cref{fig:running_1} has three coarse sets, one containing the root, another containing player $1$'s information set, and another containing all of player $2$'s information sets.

A strategy $\sigma_n$ is \textbf{coarse} if it is measurable with respect to $\mathcal{F}$, that is, if for any $I$ and $I'$ that are subsets of the same coarse set, $\sigma_n(I) = \sigma_n(I')$. Coarse strategies formalize what it means to neglect the dependence of a player's action on that player's information.

For the rest of this section, we restrict attention to finite games. Each strategy profile $\sigma$ induces a probability measure $\measure[\sigma]$ on the set of terminal histories $Z$. Given some $Q \subseteq Z$, we use $\measure[\sigma](Q)$ to denote the probability of $Q$. We extend the $\measure$ notation to sets $Q \subseteq H$ that might contain non-terminal histories, by regarding each set $Q$ as equivalent to the set of terminal histories that succeed it, $\{z \in Z: Q \preceq \{z\}\}$. For instance, $\measure[\sigma](Q \cap Q') = \measure[\sigma](\{z \in Z: Q \preceq \{z\}\} \cap \{z \in Z: Q' \preceq \{z\}\})$. Given sets $Q,Q' \subseteq H$, we use $\measure[\sigma](Q \mid Q')$ to denote the relevant conditional probability. We abuse notation, using $\measure[\overline{\sigma}^I](Q)$ to denote the probability of event $Q$ under partial strategy profile $\overline{\sigma}^I$, taking care to do this only when the relevant probability is pinned down by $\overline{\sigma}^I$.\footnote{That is, only when $\measure[\sigma](Q) = \measure[\sigma'](Q)$ for any $\sigma$ and $\sigma'$ that extend $\overline{\sigma}^I$.}

A \textbf{conjecture}\footnote{In calling these `conjectures', we are not claiming that each player explicitly believes that other players' actions are independent of their information. People are likely, if asked directly, to realize that actions are linked to information. Our theory asserts that people neglect this link in complex strategic situations, and instead adopt the heuristic of treating actions and information as independent.} for information set $I$ is a profile of partial strategies $\overline{\sigma}^I \equiv (\overline{\sigma}_n^I)_{n \in N \cup \{\nature\}}$ such that $\measure[\overline{\sigma}^I](I) > 0$. A \textbf{system of conjectures} $(\overline{\sigma}^I)_{I \in \mathcal{I}_N}$ specifies one conjecture at each information set. An \textbf{assessment} is a tuple $(\sigma_N, (\overline{\sigma}^I)_{I \in \mathcal{I}_N})$ consisting of a strategy profile $\sigma_N$ and a system of conjectures $(\overline{\sigma}^I)_{I \in \mathcal{I}_N}$.

Given information set $I$ and history $h \in I$, the conjecture $\overline{\sigma}^I$ suffices to pin down each player's continuation utility, which we denote $u_n(h,\overline{\sigma}^I)$. Given an information set $I$ and conjecture $\overline{\sigma}^I$, we define $\rho(\overline{\sigma}^I, I) \in \Delta (I)$ as the distribution over $I$ implied by Bayes' rule; that is, for each history $h \in I$, we have $\rho(h,\overline{\sigma}^I, I) = \measure[\overline{\sigma}^I](\{h\} \mid I)$.

An assessment $(\sigma_N, (\overline{\sigma}^I)_{I \in \mathcal{I}_N})$ is \textbf{locally rational} if for every player $n \in N$ and every information set $I \in \mathcal{I}_n$, player $n$'s strategy $\sigma_n$ and their self-conjecture $\overline{\sigma}^I_n$ yield the same action distribution at $I$, that is  $\sigma_n(I) = \overline{\sigma}^I_n(I)$, and furthermore $\overline{\sigma}^I_n$ maximizes $n$'s expected utility when play is initialized by $\rho(\overline{\sigma}^I,I)$ and proceeds according to $\overline{\sigma}^I_{-n}$, that is 
\begin{equation}
    \overline{\sigma}^I_n \in \argmax_{\hat{\sigma}^I_n} \left\{ \sum_{h \in I}  \rho(h, \overline{\sigma}^I,I) u_n(h, \hat{\sigma}^I_n, \overline{\sigma}_{-n}^I) \right\}.
\end{equation}
This definition of local rationality allows that player $n$ incorrectly foresees their own future actions---we have not imposed the requirement that the self-conjecture $\overline{\sigma}^I_n$ agrees with $\sigma_n$ at later information sets.\footnote{It does not matter whether $\overline{\sigma}^I_n$ and $\sigma_n$ agree at information sets that precede $I$, because we have restricted attention to games of perfect recall and initialized play at information set $I$.}  One motivation for the present study is that people reason differently about hypothetical and observed events; thus, the definition allows that players may make plans and then change them.

We now specify how player $n$'s conjectures depend on the strategies of other players and of nature. Consider the game and strategy profile in \Cref{fig:running_1}. Conditional on nature's coarse set and player $1$'s information set both being reached, nature plays $\omega_2$ with probability $\frac{1}{3}$ and $\omega_3$ with probability $\frac{2}{3}$. Conditional on player $1$'s information set and player $2$'s coarse set both being reached, $2$ plays $l$ with probability $\frac{1}{3}$ and $r$ with probability $\frac{2}{3}$. In essence, we will require that at player $1$'s information set, player $1$ understands these probabilities, but neglects how player $2$'s actions depend on what player $2$ knows, as captured by the conjecture in \Cref{fig:running_2}. We generalize this idea with the following definition. Let $(h,a)$ denote the immediate successor of history $h$ at which action $a$ was just played.

\begin{definition}\label{def:cursed_plausible}
For player $n \in N$ and information set $I \in \mathcal{I}_n$, the conjecture $\overline{\sigma}^I$ is \textbf{cursed-plausible} with strategy profile $\sigma_N$ if
\begin{enumerate}
    \item For every $I' \in \mathcal{I}_n$ that precedes $I$, the distribution $\overline{\sigma}^I_n(I')$ places probability $1$ on the action\footnote{Since $I' \prec I$, this action exists. Since we have assumed perfect recall, this action is unique.} that does not rule out reaching $I$.
    \item For every player $m \neq n$, every information set $I' \in \mathcal{I}_m$, and every action $a \in A(I')$, for $\sigma \equiv (\sigma_N,\sigma_\nature)$, we have
    \begin{equation}\label{eq:CPS_update}
        \overline{\sigma}_m^I(a,I') = \measure[\sigma](\{(h,a): h \in F(I')\} \mid I \cap F(I'))
    \end{equation}
    if $\measure[\sigma](I \cap F(I')) > 0$.
    \item For all $m \neq n$, the partial strategy  $\overline{\sigma}_m^I$ is coarse.
\end{enumerate}
We say that an assessment is cursed-plausible if the above requirements hold for every player $n \in N$ and every information set $I \in \mathcal{I}_n$.
\end{definition}

Clause 2 of \Cref{def:cursed_plausible} applies only when $\measure[\sigma](I \cap F(I')) > 0$, so it does not imply Clause 3. If the strategy profile $\sigma_N$ is fully mixed, then Clause 2 implies Clause 3.

Equation \eqref{eq:CPS_update} assesses the probability that action $a$ is played conditional on the event $I \cap F(I')$, that is, the event that $I$ is reached and $F(I')$ is reached.\footnote{Recall that under the $\measure$ notation, we treat each set as equivalent to the terminal histories that succeed it. Hence, while $I$ and $F(I')$ are disjoint sets, the event $I \cap F(I')$ denotes the set of terminal histories that succeed both $I$ and $F(I')$, that is $\{z \in Z: I \preceq \{z\}\} \cap \{z \in Z: F(I') \preceq \{z\}\}$. This set is non-empty if $I'$ is compatible with $I$.} By contrast, the Bayesian posterior probability that $a$ is played at $I'$ conditions on the event $I \cap I'$, resulting in the expression $\measure[\sigma](\{(h,a): h \in {\color{blue} I'}\} \mid I \cap {\color{blue} I'})$. For example, in \Cref{fig:running_1}, if we let $I = \{\omega_2, \omega_3\}$ and $I' = \{\omega_2 y\}$, then  Equation \eqref{eq:CPS_update} conditions on the event $\{\omega_2 y, \omega_3 y\}$. Generally, Equation \eqref{eq:CPS_update} results in a coarse partial strategy for player $m$, that captures the idea that player $n$ understands the conditional probability of each player $m$ action, but neglects how $m$'s actions depends on $m$'s information.   An important subtlety is that \eqref{eq:CPS_update} also constrains nature's partial strategy $\overline{\sigma}_\nature$, to account for the fact that information set $I$ is reached.

Given the strategy profile in \Cref{fig:running_1}, we depict a cursed-plausible conjecture at player $1$'s information set in \Cref{fig:running_2}. According to this conjecture, player $2$ has a coarse strategy $\left(l \ \frac{1}{3}, r \ \frac{2}{3}\right)$ that matches the empirical frequencies of each action. This is distinct from the unweighted average, $\left(l \ \frac{1}{2}, r \ \frac{1}{2}\right)$, because nature plays $\omega_3$ twice as often as $\omega_2$.\footnote{Using the unweighted average would result in a solution concept that varies if we split states into several identical states, dividing the probabilities between them.} It is also distinct from the \textit{ex ante} probability of each action $\left(l \ \frac{3}{5}, r \ \frac{2}{5}\right)$ , because we are conditioning on player $1$'s information set.\footnote{If we applied ABEE \citep{jehiel2005analogy} to \Cref{fig:running_1} using the coarse sets as analogy classes, player $1$ would best-respond to the \textit{ex ante} probability $\left(l \ \frac{3}{5}, r \ \frac{2}{5}\right)$.}

Cursed-plausibility implies that players learn from observed actions. For instance, consider the sequential trading game in \Cref{fig:trading_seq}, and let $\sigma$ be a strategy profile such that player $1$ accepts at information set $\{\omega_1,\omega_2\}$ and declines at information set $\{\omega_3\}$. Let $I_\lambda$ denote nature's information set and let $I_2 = \{\omega_2 a, \omega_3 a\}$ denote player $2$'s right information set. By cursed-plausibility, at $I_2$ player $2$ conjectures that nature has played $\omega_2$ for sure, since
    \begin{equation}
       \measure[\sigma](\{(h,\omega_2): h \in F(I_\lambda)\} \mid I_2 \cap F(I_\lambda)) = 1.
    \end{equation}

Cursed-plausibility assumes that players have, at each information set, correct beliefs about the distribution of other players' actions, but neglect how each player's actions depend on what other players (or nature) have done. One learning story might justify this unusual assumption, that we sketch informally: Suppose that many generations of players play a game. At each information set $I$, the active player consults statistics about past behavior, restricting attention to plays that passed through $I$. But rather than responding to the joint distribution of other players' actions, which might be high-dimensional\footnote{Consider a player keeping track of the play of $j$ other players; each has one coarse set with $k$ actions. The joint distribution of the other players' actions can occupy an array as large as $k^j$, whereas the marginal distributions require an array of size no more than than $kj$.}, he adopts the heuristic of treating behavior at each coarse set as independent of behavior at other coarse sets. Crucially, the player crosses each bridge as he comes to it; he consults only the statistics for the current information set.\footnote{As \cite{savage1972foundations} observes, in complex decision problems, ``there are innumerable bridges one cannot afford to cross, unless he happens to come to them."}


An assessment $(\sigma_N, (\overline{\sigma}^I)_{I \in \mathcal{I}_N})$ is a \textbf{weak perfect cursed equilibrium} (WPCE) if it is locally rational and cursed-plausible.

For example, there exists a WPCE of \Cref{fig:running_1} with the strategy profile in blue. At player $1$'s information set, playing $y$ is locally rational given the conjecture in \Cref{fig:running_2}; since player $1$ believes that $2$ plays $r$ with probability $\frac{2}{3}$ independently of $2$'s information, the perceived payoff from $y$ is $\frac{1}{3} \cdot \frac{2}{3} \cdot 6$, which strictly exceeds the payoff of $1$ from $x$. The cursed-plausible conjectures for player $2$'s information sets are trivial, and $2$'s strategy is locally rational.

WPCE extends to games with infinite action sets; the details are in \Cref{sec:WPCE_infinite}.

Cursed-plausibility does not restrict conjectures conditional on zero-probability events. Returning to the running example, consider the strategy profile in \Cref{fig:running_3}; because player $1$ plays $x$ for sure, Clause $2$ of \Cref{def:cursed_plausible} does not pin down what player $1$ believes about player $2$. Consequently, the conjecture depicted in \Cref{fig:running_4} is cursed-plausible. Facing that conjecture, it is locally rational for player $1$ to play $x$, so \Cref{fig:running_join_3_4}, together with the trivial conjectures for player $2$, constitutes a WPCE.

\begin{figure}[h]
\centering
\begin{subfigure}[T]{.48\textwidth}
  \centering
\begin{tikzpicture}[font=\footnotesize,edge from parent/.style={draw,thick}]
    \tikzstyle{solid}=[circle,draw,inner sep=1.2,fill=black];
    \tikzstyle{hollow}=[circle,draw,inner sep=1.2]
    \tikzstyle{solid_red}=[circle,draw=red,inner sep=1.7,fill=red];
    \tikzstyle{solid_mag}=[circle,draw=magenta,inner sep=1.7,fill=red];
    \tikzstyle{level 1}=[level distance=18mm,sibling distance=24mm]
    \tikzstyle{level 2}=[level distance=12mm,sibling distance=12mm]
    \tikzstyle{level 3}=[level distance=12mm,sibling distance=12mm]
     \node(0)[hollow]{}
     child[level distance=30mm,sibling distance=20mm]{node[solid]{}
       child{ node[below]{}
        edge from parent node[left]{$l$}
        }
       child{node[below]{} 
        edge from parent node[right]{$r$}
        }
     }
     child{node[solid]{}
       child{node[below]{} 
        edge from parent node[left]{$x$}
        }
        child{node[solid]{}
        child{node[below]{} edge from parent node[left]{$l$}}
        child{node[below]{} edge from parent node[right]{$r$}}
        edge from parent node[right]{$y$}
        }
    edge from parent node[xshift=.5mm, name=omega2]{$\omega_2 \ {\color{blue} [.2]}$}
     }
     child{node[solid]{}
       child{node[below]{} 
        edge from parent node[left]{$x$}
        }
       child{node[solid]{} 
        child{node[below]{} edge from parent node[left]{$l$}}
        child{node[below]{} edge from parent node[right]{$r$}}
        edge from parent node[right]{$y$}
        }
    edge from parent node[right, xshift=2mm]{$\omega_3{\color{blue} [.4]}$}
     }
     ;
     \node[left=.5mm of omega2] {$\omega_1 {\color{blue} [.4]}$};
      \draw[dashed](0-2)to(0-3);
     \node[above,yshift=2]at(0){nature}; \node[above,yshift=2]at(0-1){2};
     \node[above right,yshift=2]at(0-2){1};
     \node[above,yshift=2]at(0-3){1};
     \foreach \x in {2,3} \node[above right,yshift=2]at(0-\x-2){2};
    \node[below]at(0-1-1){\stackon{$1$}{$0$}};
    \node[below]at(0-1-2){\stackon{$0$}{$0$}};
    \node[below]at(0-2-1){\stackon{$0$}{$1$}};
    \node[below]at(0-3-1){\stackon{$0$}{$1$}};
    \node[below]at(0-2-2-1){\stackon{$1$}{$0$}};
    \node[below]at(0-2-2-2){\stackon{$0$}{$6$}};
    \node[below]at(0-3-2-1){\stackon{$0$}{$0$}};
    \node[below]at(0-3-2-2){\stackon{$1$}{$0$}};
    \draw[blue, ->, line width=1.5pt, -{Triangle[length=8pt]},] (0-1) -- ($(0-1)!.97!(0-1-1)$);
    \draw[blue, ->, line width=1.5pt, -{Triangle[length=8pt]},] (0-2) -- ($(0-2)!.97!(0-2-1)$);
    \draw[blue, ->, line width=1.5pt, -{Triangle[length=8pt]},] (0-3) -- ($(0-3)!.97!(0-3-1)$);
    \draw[blue, ->, line width=1.5pt, -{Triangle[length=8pt]},] (0-2-2) -- ($(0-2-2)!.97!(0-2-2-1)$);
    \draw[blue, ->, line width=1.5pt, -{Triangle[length=8pt]},] (0-3-2) -- ($(0-3-2)!.97!(0-3-2-2)$);
\end{tikzpicture}
  \caption{Extensive game, with strategy profile in blue.}
  \label{fig:running_3}
\end{subfigure}
\begin{subfigure}[T]{.48\textwidth}
  \centering
\begin{tikzpicture}[font=\footnotesize,edge from parent/.style={draw,thick}]
    \tikzstyle{solid}=[circle,draw,inner sep=1.2,fill=black];
    \tikzstyle{hollow}=[circle,draw,inner sep=1.2]
    \tikzstyle{solid_red}=[circle,draw=red,inner sep=1.7,fill=red];
    \tikzstyle{solid_mag}=[circle,draw=magenta,inner sep=1.7,fill=red];
    \tikzstyle{level 1}=[level distance=18mm,sibling distance=24mm]
    \tikzstyle{level 2}=[level distance=12mm,sibling distance=12mm]
    \tikzstyle{level 3}=[level distance=12mm,sibling distance=12mm]
     \node(0)[hollow]{}
     child[level distance=30mm,sibling distance=20mm]{node[solid]{}
       child{ node[below]{}
        edge from parent node[left]{$l$}
        }
       child{node[below]{} 
        edge from parent node[right]{$r$}
        }
     }
     child{node[solid_red]{}
       child{node[below]{} 
        edge from parent node[left]{$x$}
        }
        child{node[solid]{}
        child{node[below]{} edge from parent node[left]{$l$}}
        child{node[below]{} edge from parent node[right]{$r$}}
        edge from parent node[right]{$y$}
        }
    edge from parent node[xshift=.7mm, name=omega2]{$\omega_2 \ {\color{blue} \left[\frac{1}{3}\right]}$}
     }
     child{node[solid_red]{}
       child{node[below]{} 
        edge from parent node[left]{$x$}
        }
       child{node[solid]{} 
        child{node[below]{} edge from parent node[left]{$l$}}
        child{node[below]{} edge from parent node[right]{$r$}}
        edge from parent node[right]{$y$}
        }
    edge from parent node[right, xshift=2mm]{$\omega_3 {\color{blue} \left[\frac{2}{3}\right]}$}
     }
     ;
     \node[left=.5mm of omega2] {$\omega_1 {\color{blue} [0]}$};
      \draw[dashed, red, line width=1.5](0-2)to(0-3);
     \node[above,yshift=2]at(0){nature}; \node[above,yshift=2]at(0-1){2};
     \node[above right,yshift=2]at(0-2){1};
     \node[above,yshift=2]at(0-3){1};
     \foreach \x in {2,3} \node[above right,yshift=2]at(0-\x-2){2};
     \node[below]at(0-1-1){\stackon{${\color{white}1}$}{${\color{white}0}$}};
     \node[below]at(0-1-2){\stackon{${\color{white}0}$}{${\color{white}0}$}};
     \node[below]at(0-2-1){\stackon{${\color{white}0}$}{${\color{white}1}$}};
     \node[below]at(0-3-1){\stackon{${\color{white}0}$}{${\color{white}1}$}};
     \node[below]at(0-2-2-1){\stackon{${\color{white}1}$}{${\color{white}0}$}};
     \node[below]at(0-2-2-2){\stackon{${\color{white}0}$}{${\color{white}6}$}};
     \node[below]at(0-3-2-1){\stackon{${\color{white}0}$}{${\color{white}0}$}};
     \node[below]at(0-3-2-2){\stackon{${\color{white}1}$}{${\color{white}0}$}};
    \draw[blue, ->, line width=1.5pt, -{Triangle[length=8pt]},] (0-2) -- ($(0-2)!.97!(0-2-1)$);
    \draw[blue, ->, line width=1.5pt, -{Triangle[length=8pt]},] (0-3) -- ($(0-3)!.97!(0-3-1)$);
    \draw[blue, ->, line width=1.5pt, -{Triangle[length=8pt]},] (0-2-2) -- ($(0-2-2)!.97!(0-2-2-1)$);
    \draw[blue, ->, line width=1.5pt, -{Triangle[length=8pt]},] (0-3-2) -- ($(0-3-2)!.97!(0-3-2-1)$);
\end{tikzpicture}
  \caption{Cursed-plausible conjecture at player $1$'s information set.}
  \label{fig:running_4}
\end{subfigure}%
\caption{Another WPCE of the running example.}\label{fig:running_join_3_4}
\end{figure}
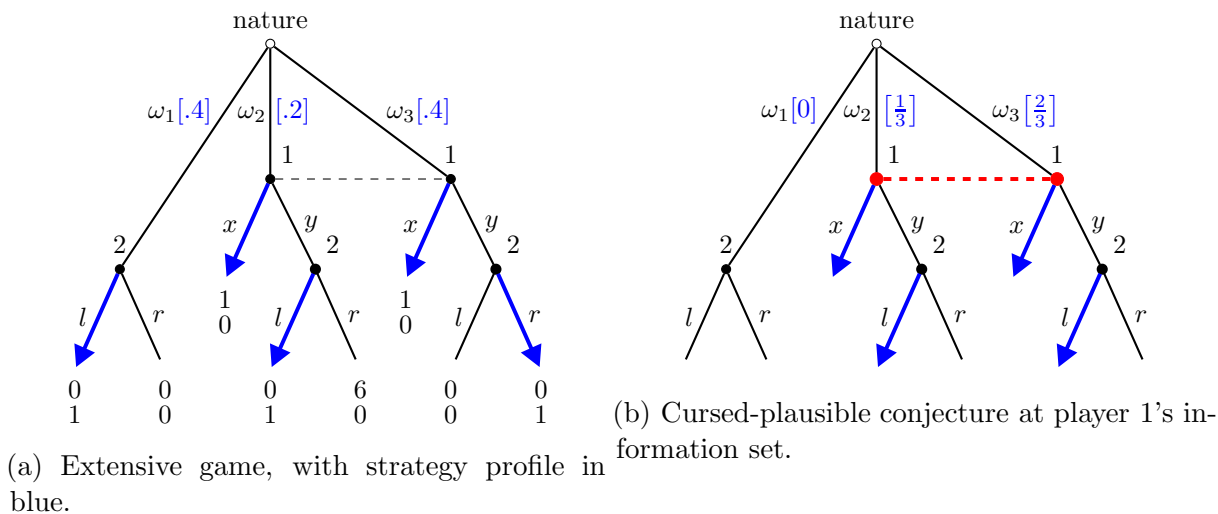

The conjecture in \Cref{fig:running_4} seems farfetched, because player $1$'s conjecture about player $2$ is not anchored by player $2$'s actual behavior. The problem is not just that cursed-plausibility does not discipline off-path beliefs; rather, it is that cursed-plausibility does not discipline on-path forecasts about off-path actions.\footnote{Weak perfect Bayesian equilibrium implicitly assumes that player $1$ correctly forecasts player $2$'s off-path actions. But at the strategy profile in \Cref{fig:running_3} the correct forecast is not a coarse strategy. Our problem, therefore, is to construct a conjecture about player $2$'s behavior that is coarse but not arbitrary.}

For some applications, such as those studied in \Cref{sec:applications}, these subtleties about beliefs do not matter for behavior, and WPCE is enough for practical purposes.

To deal with harder cases, we propose a requirement that strengthens cursed-plausibility to restrict conjectures involving zero-probability events. The idea is to use trembles to specify conditional probabilities, as in \cite{kreps1982sequential}.  Observe that if the strategy profile $\sigma_N$ is fully mixed, then for any player $n$ and any compatible information sets $I, I' \in \mathcal{I}_n$, we have $\measure[(\sigma_N,\sigma_\nature)](I \cap F(I')) > 0$, so cursed-plausibility pins down player $n$'s conjecture about other players and nature by Equation \eqref{eq:CPS_update}.

An asssessment $(\sigma_N, (\overline{\sigma}^I)_{I \in \mathcal{I}_N})$ is \textbf{cursed-consistent} if there exists a sequence of cursed-plausible assessments $\left(\sigma^k_N, (\overline{\sigma}^{I,k})_{I \in \mathcal{I}_N}\right)_{k \in \mathbb{N}}$ converging to $(\sigma_N, (\overline{\sigma}^I)_{I \in \mathcal{I}_N})$  such that for all $k$, the strategy profile $\sigma^k_N$ is fully mixed. For example, given the strategy profile in \Cref{fig:running_1}, cursed-consistency allows the conjecture in \Cref{fig:running_2}. But given the strategy profile in \Cref{fig:running_3}, cursed-consistency rules out the conjecture in \Cref{fig:running_4}.

\begin{theorem}\label{thm:cursed_consistent_implies_plausible}
    For any finite game, every cursed-consistent assessment is cursed-plausible.
\end{theorem}

An assessment $(\sigma_N, (\overline{\sigma}^I)_{I \in \mathcal{I}_N})$ is a \textbf{sequential cursed equilibrium} (SCE) if it is locally rational and cursed-consistent.

Every SCE is a WPCE, by \Cref{thm:cursed_consistent_implies_plausible}. But not every WPCE is a SCE; for instance, the WPCE in \Cref{fig:running_join_3_4} is not a SCE. The running example has a unique SCE, as in \Cref{fig:running_join_1_2}.

Next we state an existence theorem.

\begin{theorem}\label{thm:SCE_existence}
    For any finite game, there exists a sequential cursed equilibrium.
\end{theorem}

We sketch the proof, deferring details to the appendix. Consider the perturbed game that constrains each action to be played with at least probability $\epsilon$, for parameter $\epsilon > 0$. By Kakutani's theorem, each $\epsilon$-perturbed game has a cursed-plausible assessment that is constrained locally rational. We then take $\epsilon$ to $0$; by compactness we pass to a convergent subsequence, with limit assessment $(\sigma_N, (\overline{\sigma}^I)_{I \in \mathcal{I}_N})$. We prove that the resulting limit has a well-defined system of conjectures, that is $\measure[\overline{\sigma}^I](I) > 0$ for all $I \in \mathcal{I}_N$. The limit assessment $(\sigma_N, (\overline{\sigma}^I)_{I \in \mathcal{I}_N})$ is cursed-consistent by construction, and is locally rational by upper hemicontinuity of the best-response correspondence.

\begin{corollary}
    For any finite game, there exists a weak perfect cursed equilibrium.
\end{corollary}

Next we state that SCE has the same continuity properties as sequential equilibrium.\footnote{See Proposition 2 of \cite{kreps1982sequential}.}

\begin{proposition}\label{prop:SCE_uhc}
Fixing an extensive form, the correspondence from pairs $((u_n)_{n \in N},\sigma_\nature)$ of utilities and fully mixed nature strategies to the set of SCE for the game so defined is upper hemicontinuous.
\end{proposition}
We omit the proof of \Cref{prop:SCE_uhc} because it follows from the usual arguments.

\begin{example}\normalfont SCE allows cursedness about endogenous information. The game in \Cref{fig:SCE_starter} has a unique sequential equilibrium, in which $1$ plays $L$, $2$ plays $l$ after $L$ and $r$ after $R$, and $3$ plays $a$. Let us denote the information sets $I_1$, $I_2^L$, $I_2^R$, and $I_3$ respectively. Observe that the coarse sets are $\mathcal{F} = \left\{I_1, I_2^L \cup I_2^R, I_3 \right\}$. We construct a SCE in which player $1$ mixes. Let $p$ be the probability that $1$ plays $L$. By local rationality, in every SCE player $2$ plays $l$ after $L$ and $r$ after $R$.  Thus, player $3$ conjectures that $1$ plays $L$ with probability $p$ and that $2$ adopts a coarse strategy that plays $l$ with probability $p$. Here, mixing by player $1$ causes player $3$ to make a mistake. From $3$'s perspective, player $1$ plays $L$ and $R$ both with positive probability, and player $2$ plays $l$ and $r$ both with positive probability. But player $3$ neglects the correlation between the actions of $1$ and $2$. In our notation, the conjecture $\overline{\sigma}^{I_3}$ specifies $\overline{\sigma}_1^{I_3}(L,I_1)=p$ and $\overline{\sigma}_2^{I_3}(l,I_2^L)=\overline{\sigma}_2^{I_3}(l,I_2^R)=p$. There exists a SCE such that $p$ solves $2p(1-p)3=1$ and player $3$ plays $a$ with probability $\frac{1}{2}$.
\end{example}

\begin{figure}[h]
    \centering
\begin{tikzpicture}[font=\footnotesize,edge from parent/.style={draw,thick}]
    \tikzstyle{solid}=[circle,draw,inner sep=1.2,fill=black];
    \tikzstyle{hollow}=[circle,draw,inner sep=1.2]
    \tikzstyle{solid_red}=[circle,draw=red,inner sep=1.7,fill=red];
    \tikzstyle{solid_mag}=[circle,draw=magenta,inner sep=1.7,fill=red];
    \tikzstyle{level 1}=[level distance=5mm,sibling distance=64mm]
    \tikzstyle{level 2}=[level distance=10mm,sibling distance=32mm]
    \tikzstyle{level 3}=[level distance=10mm,sibling distance=16mm]
     \node(0)[hollow]{}
     child{node[solid]{}
        child{node[solid]{}
            child{node[below]{}
            edge from parent node[left]{$a$}
            }
            child{node[below]{}
            edge from parent node[right]{$b$}
            }
        edge from parent node[above left]{$l$}
        }
        child{node[solid]{}
            child{node[below]{}
            edge from parent node[left]{$a$}
            }
            child{node[below]{}
            edge from parent node[right]{$b$}
            }
        edge from parent node[above right]{$r$}
        }
     edge from parent node[above left]{$L$}
     }
     child{node[solid]{}
        child{node[solid]{}
            child{node[below]{}
            edge from parent node[left]{$a$}
            }
            child{node[below]{}
            edge from parent node[right]{$b$}
            }
        edge from parent node[above left]{$l$}
        }
        child{node[solid]{}
            child{node[below]{}
            edge from parent node[left]{$a$}
            }
            child{node[below]{}
            edge from parent node[right]{$b$}
            }
        edge from parent node[above right]{$r$}
        }
     edge from parent node[above right]{$R$}
     }
     ;

    \node[below]at(0-1-1-1){\stackon{\stackon{$1$}{$1$}}{$1$}};
    \node[below]at(0-1-1-2){\stackon{\stackon{$0$}{$1$}}{$1$}};
    \node[below]at(0-1-2-1){\stackon{\stackon{$1$}{$0$}}{$1$}};
    \node[below]at(0-1-2-2){\stackon{\stackon{$3$}{$0$}}{$1$}};
    \node[below]at(0-2-1-1){\stackon{\stackon{$1$}{$0$}}{$0$}};
    \node[below]at(0-2-1-2){\stackon{\stackon{$3$}{$0$}}{$2$}};
    \node[below]at(0-2-2-1){\stackon{\stackon{$1$}{$1$}}{$0$}};
    \node[below]at(0-2-2-2){\stackon{\stackon{$0$}{$1$}}{$2$}};
   \draw[dashed](0-1-1)to(0-2-2);
     \node[above,yshift=2]at(0){1};
     \node[above,yshift=2]at(0-1){2}; 
     \node[above,yshift=2]at(0-2){2}; 
     \node[above,yshift=2]at(0-1-1){3};
     \node[above,yshift=2]at(0-1-2){3};
     \node[above,yshift=2]at(0-2-1){3};
     \node[above,yshift=2]at(0-2-2){3};
     \node[below]at($(0-1-1)!.5!(0-2-2)$){$I_3$};
\end{tikzpicture}
    \caption{}
    \label{fig:SCE_starter}
\end{figure}

\begin{example}\normalfont \label{ex:groucho}
SCE players can be mistaken about their own future actions, and these mistakes can be costly. For instance, \Cref{fig:trading_confirm} depicts a game in which Groucho can join a club. If Groucho and the club both accept, then Groucho can confirm membership or resign at a cost. There are five information sets, and four coarse sets: the singleton set with the root and the sets $\{\omega_1,\omega_2\}$, $\{\omega_1a,\omega_2a\}$, and $\{\omega_1 a a, \omega_2 a a\}$.
\end{example}
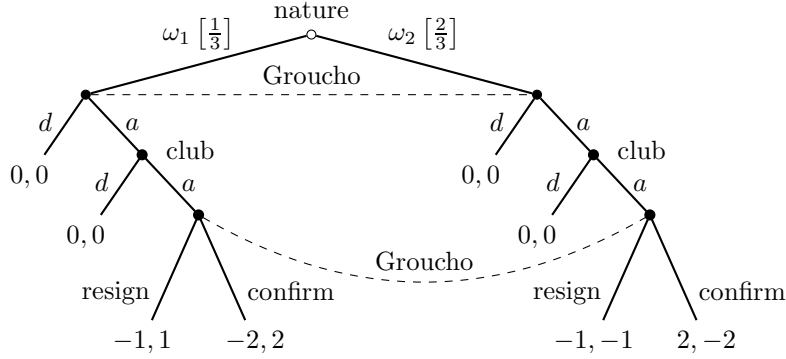
\begin{figure}[h]
    \centering
\begin{tikzpicture}[font=\footnotesize,edge from parent/.style={draw,thick}]
    \tikzstyle{solid}=[circle,draw,inner sep=1.2,fill=black];
    \tikzstyle{hollow}=[circle,draw,inner sep=1.2]
    \tikzstyle{level 1}=[level distance=8mm,sibling distance=60mm]
    \tikzstyle{level 2}=[level distance=8mm,sibling distance=15mm]
    \tikzstyle{level 3}=[level distance=8mm,sibling distance=15mm]
    \tikzstyle{level 4}=[level distance=14mm,sibling distance=15mm]
     \node(0)[hollow]{}
     child{node[solid]{}
        child{node[below]{$0,0$}
        edge from parent node[left, yshift=2]{$d$}
        }
       child{node[solid]{}
       child{node[below]{$0,0$} edge from parent node[left, yshift=2]{$d$}}
        child{node[solid]{}
            child{node[below]{$-1,1$}
                edge from parent node[below left, xshift=-5]{resign}
            }
            child{node[below]{$-2,2$}
                edge from parent node[below right, xshift=5]{confirm}
            }
        edge from parent node[right]{$a$}}
        edge from parent node[right]{$a$}
        }
    edge from parent node[above, yshift = .5mm]{$\omega_1 \left[\frac{1}{3}\right]$}
     }
child{node[solid]{}
        child{node[below]{$0,0$}
        edge from parent node[left, yshift=2]{$d$}
        }
       child{node[solid]{}
       child{node[below]{$0,0$} edge from parent node[left, yshift=2]{$d$}}
        child{node[solid]{}
            child{node[below]{$-1,-1$}
                edge from parent node[below left, xshift=-5]{resign}
            }
            child{node[below]{$2,-2$}
                edge from parent node[below right, xshift=5]{confirm}
            }
        edge from parent node[right]{$a$}}
        edge from parent node[right]{$a$}
        }
    edge from parent node[above, yshift=.5mm]{$\omega_2 \left[\frac{2}{3}\right]$}
     }
     ;
    \draw[dashed](0-1)to(0-2);
    \draw[dashed,bend right=30](0-1-2-2)to(0-2-2-2);
     \node[above,yshift=2]at(0){nature};
     \node[right,yshift=2,xshift=5]at(0-1-2){club};
     \node[right,yshift=2,xshift=5]at(0-2-2){club};
     \node[above]at($(0-1)!.5!(0-2)$){Groucho};
     \node[above,yshift=-25]at($(0-1-2-2)!.5!(0-2-2-2)$){Groucho};
\end{tikzpicture}
    \caption{Payoffs are written ``$x,y$", where $x$ is Groucho's payoff and $y$ is the club's payoff.}
    \label{fig:trading_confirm}
\end{figure}

\Cref{ex:groucho} has a unique SCE. Local rationality implies that the club accepts at history $\omega_1 a$ and declines at $\omega_2 a$. Given this behavior, at Groucho's first information set any cursed-consistent\footnote{This argument relies on cursed-consistency; cursed-plausibility does not suffice to pin down Groucho's conjecture about the club if Groucho plays $d$. There exists a WPCE in which Groucho plays $d$, conjecturing that the club will always play $d$.} conjecture specifies that nature plays $\left(\omega_1\ \frac{1}{3}, \ \omega_2 \ \frac{2}{3}\right)$, and that the club plays a coarse strategy that accepts with probability $\frac{1}{3}$. It is locally rational for Groucho to accept, planning to play \textit{confirm}, which results in the expected payoff $\frac{1}{9}\cdot (-2)+\frac{2}{9}\cdot 2$. However, at Groucho's later information set, the cursed-consistent conjecture specifies that nature plays $\omega_1$ and that the club plays $a$, so it is locally rational to resign.\footnote{At the later information set, Equation \eqref{eq:CPS_update} conditions on the event $\{\omega_1 a a,\ \omega_2 a a\}$. By contrast, ABEE with analogy partition $\mathcal{F}$ would require Groucho to respond, at both information sets, as if the club was playing the coarse strategy that accepts with probability $\frac{1}{3}$. Thus the unique SCE strategy profile cannot be supported under ABEE, demonstrating that SCE is not a refinement of ABEE or \textit{vice versa}.} Groucho Marx reportedly wrote, ``Please accept my resignation. I don't want to belong to any club that would accept me as a member."

In practice, observing events might help players reason hypothetically about similar events. Suppose that we extend \Cref{fig:trading_confirm} so that after resigning, Groucho has the opportunity to join a second club, whose value is drawn independently of the first club. Plausibly, Groucho will realize \textit{ex post} that he made a mistake when reasoning hypothetically about the first club, and thus decline to join the second club. SCE does not capture this kind of learning. How much players learn, and what counts as a similar event, is an empirical question.\footnote{In their experiment, \cite{esponda2014hypothetical} find that providing 40 rounds of detailed feedback increases the share of strategic voters to 50\%, compared to 22\% in a no-feedback treatment. They also run a no-feedback sequential treatment that does not require hypothetical thinking, which results in 76\% of subjects voting strategically. However, when subjects with 90 rounds of experience with sequential voting are then asked to vote simultaneously, the proportion of strategic play falls by about 54 percentage points.}

One might conjecture that SCE coincides with sequential equilibrium if every coarse set contains exactly one information set, that is, if $\mathcal{F} = \mathcal{I}$. This conjecture is false; \Cref{sec:FI_non_equiv} states a counterexample. In essence, even if $\mathcal{F} = \mathcal{I}$, it is still possible that the cursed-plausible conjectures neglect to condition on strategically-relevant hypothetical events.

\subsection{Simultaneous Bayesian games}\label{sec:simultaneous_bayesian}

We shed light on the relationship between SCE and cursed equilibrium, by considering a special class of games. In a \textbf{simultaneous Bayesian game}, nature moves just once, at the initial history, selecting a state $\omega \in \Omega$. Then each player $n \in N$ observes their \textbf{type}, which is a cell of some partition $\mathcal{T}_n$ of $\Omega$. Then all players move once and at the same time. Formally, we require that for any player $n \in N$ and any histories $h, h' \in H_n$, $h$ and $h'$ are in the same information set if and only if there exists type $T_n \in \mathcal{T}_n$ and $\omega, \omega' \in T_n$ such that $\omega \preceq h$ and $\omega' \preceq h'$. 

We additionally assume that player $n$'s available actions $A_n$ are the same at all information sets in $\mathcal{I}_n$, so that $F(I) = \cup \mathcal{I}_n$ for all $I \in \mathcal{I}_n$. Abusing notation, we write a strategy as a function of types rather than a function of information sets, so a strategy for player $n$ is a function $\sigma_n: A_n \times \mathcal{T}_n \rightarrow [0,1]$, satisfying $\sum_{a_n \in A_n} \sigma_n(a_n,T_n) = 1$.  We define  $\sigma_n(a_n, \omega) \equiv \{\sigma_n(a_n, T_n) : \omega \in T_n\}$.  The probability that nature plays state $\omega$ is denoted $\sigma_\nature(\omega)$.

We now define cursed equilibrium for simultaneous Bayesian games. Let $\beta(a_{-n} \mid \sigma_{-n}, T_n)$ be the probability that opponent action profile $a_{-n}$ is played, under strategy profile $\sigma_{-n}$, when player $n$'s type is $T_n$. That is,
\begin{equation}\label{eq:static_CE_beliefs}
    \beta(a_{-n} \mid \sigma_{-n}, T_n) \equiv \frac{\sum_{\omega \in T_n} \sigma_\nature(\omega)\prod_{m \in N \setminus \{n\} } \sigma_m(a_m,\omega) }{\sum_{\omega \in T_n} \sigma_\nature(\omega)}.
\end{equation}
A strategy profile $\sigma_N$ is a \textbf{cursed equilibrium (CE)} if for each player $n$ and type $T_n$, we have
\begin{equation}\label{eq:static_CE_BR}
    \sigma_n(T_n) \in \argmax_{d_n \in \Delta (A_n)} \left\{ \sum_{\omega \in T_n} \sigma_\nature(\omega) \sum_{a_{-n}} \beta(a_{-n} \mid \sigma_{-n}, T_n) u_n(d_n, a_{-n},\omega)\right\}.
\end{equation}
Under CE, each player responds as if the opponent action profiles are independent of the state, since the function $\beta$ in expression \eqref{eq:static_CE_BR} does not depend on $\omega$. However, each player may believe that their opponents' actions are correlated with each other, because expression \eqref{eq:static_CE_beliefs} captures the distribution over opponent action \textit{profiles} conditional on $T_n$.

Consider a variation of cursed equilibrium, such that each player believes that their opponents choose independently of each other. Let $\gamma_m(a_{m} \mid \sigma_m, T_n)$ be the probability that opponent $m$ plays action $a_m$, under strategy $\sigma_m$, when player $n$'s type is $T_n$, that is
\begin{equation}
    \gamma_m(a_{m} \mid \sigma_m, T_n) \equiv \frac{\sum_{\omega \in T_n} \sigma_\nature(\omega) \sigma_m(a_m,\omega) }{\sum_{\omega \in T_n} \sigma_\nature(\omega)}.
\end{equation}
A strategy profile $\sigma_N$ is an \textbf{independently cursed equilibrium (ICE)} if for each player $n$ and type $T_n$, we have
\begin{equation}
    \sigma_n(T_n) \in \argmax_{d_n \in \Delta (A_n)} \left\{ \sum_{\omega \in T_n} \sigma_\nature(\omega) \sum_{a_{-n}}  \prod_{m \in N \setminus \{n\}} \gamma_m(a_{m} \mid \sigma_m, T_n) u_n(d_n, a_{-n},\omega)\right\}.
\end{equation}

We contrast CE and ICE. CE posits that player $n$ understands the distribution over opponent \textit{action profiles} conditional on type $T_n$, but responds as though those action profiles are independent of the state. ICE posits that player $n$ understands the distribution over each opponent's \textit{actions} conditional on type $T_n$, but responds as though those actions are independent of the state \textit{and of each other}. CE and ICE are equivalent for two-player games.

Next we state that for simultaneous Bayesian games, ICE, WPCE, and SCE all predict the same behavior.
\begin{theorem}\label{thm:ICE_SCE}
    For any finite simultaneous Bayesian game, these statements are equivalent:
    \begin{enumerate}
        \item Strategy profile $\sigma_N$ is an ICE.
        \item There exists a system of conjectures $(\overline{\sigma}^I)_{I \in \mathcal{I}_N}$ such that $(\sigma_N, (\overline{\sigma}^I)_{I \in \mathcal{I}_N})$ is a WPCE.
        \item There exists a system of conjectures $(\overline{\sigma}^I)_{I \in \mathcal{I}_N}$ such that $(\sigma_N, (\overline{\sigma}^I)_{I \in \mathcal{I}_N})$ is a SCE.
    \end{enumerate}
\end{theorem}

\begin{corollary}\label{corr:CE_SCE_equiv}
    For any two-player finite simultaneous Bayesian game, CE, WPCE, and SCE are equivalent.
\end{corollary}

\Cref{thm:ICE_SCE} indicates that CE and SCE substantially agree for simultaneous Bayesian games. They coincide entirely for games with two players, and for games with three or more players, they differ only in whether each player understands the correlation between other players' actions. 

\subsection{A generalization with partial cursedness}\label{sec:chi_psi_SCE}

We now generalize SCE to allow that each player {partially} understands how the other players' actions depend on their information, using a parameter $\chi \in [0,1]$ to scale how cursed the players are. This is analogous to $\chi$-cursed equilibrium in \cite{eyster2005cursed}.

Formally, a \textbf{sequential assessment} is a tuple $(\sigma_N, (\sigma_N^k)_{k=1}^\infty)$ consisting of a strategy profile $\sigma_N$ and a sequence of fully mixed strategy profiles $(\sigma_N^k)_{k=1}^\infty$ converging to $\sigma_N$. Recall that given conjecture $\overline{\sigma}^I$, the implied distribution over histories in $I$ is denoted $\rho(\overline{\sigma}^I,I)$. We say that the sequential assessment $(\sigma_N, (\sigma_N^k)_{k=1}^\infty)$ is a \textbf{$\chi$-sequential cursed equilibrium} ($\chi$-SCE) if for every player $n \in N$ and every information set $I \in \mathcal{I}_n$, there exists a partial strategy $\tilde{\sigma}^I_n$ such that $\tilde{\sigma}^I_n(I) = \sigma_n(I)$ and $\tilde{\sigma}^I_n$ maximizes $n$'s expected utility when
\begin{enumerate}
    \item with probability $\chi$, play is initialized by the distribution $\lim_{k \rightarrow \infty}\rho(\overline{\sigma}^{I,k},I)$ and proceeds according to $\lim_{k \rightarrow \infty} \overline{\sigma}^{I,k}_{-n}$, where the conjecture $\overline{\sigma}^{I,k}  $ is cursed-plausible with strategy profile $\sigma_N^k$,
    \item and with probability $1 - \chi$,  play is initialized by the distribution $\lim_{k \rightarrow \infty} \rho((\sigma_N^{k},\sigma_\lambda),I)$ and proceeds according to $\sigma_{-n} = (\sigma_{N\setminus n}, \sigma_\lambda)$, \textit{i.e.}\ the Bayesian conjecture,
\end{enumerate} 
and all the above limits exist.

The solution concept $\chi$-SCE provides a way to interpolate between sequential equilibrium ($\chi = 0$) and sequential cursed equilibrium ($\chi = 1$).

\begin{theorem}\label{thm:chi_psi_exists}
For any finite game and any $\chi \in [0,1]$, there exists a $\chi$-SCE.
\end{theorem}

We omit the proof, since it is a straightforward adaptation of the proof of \Cref{thm:SCE_existence}.

In our view, $\chi$-SCE buys tractability at the price of some psychological realism. Calculating best-responses under $\chi$-SCE requires us to solve for the Bayesian conjecture, yet players partially neglect the link between actions and information. Thus, we interpret $\chi$-SCE as an as-if model; the calculations required to solve the model do not literally describe the player's reasoning process.


\section{Evidence from experiments}\label{sec:applications}

We provide evidence for SCE and WPCE, drawing on data from laboratory experiments. In particular, we apply these solution concepts to the pivotal voting game of \cite{esponda2014hypothetical}, the learning-from-prices game of \cite{ngangoue2021learning}, and the common-value auctions of \cite{moser2019hypothetical}. Each of these experiments uses a manipulation, such as switching from simultaneous decisions to sequential decisions, that has no effect on the Bayesian Nash equilibrium (BNE) prediction and the CE prediction. Contrary to those predictions, the manipulations substantially change subject behavior.

We find that SCE and WPCE make predictions that broadly accord with the data. These solution concepts correctly predict some errors in contingent reasoning; namely, that people make better inferences from observed events than from hypothetical events.

Many lab experiments investigate strategic behavior by having a human subject interact with computer players. In such experiments, the instructions typically describe the computer player's information and actions separately, followed by the computer player's strategy. We extend our solution concepts  to computer players in the natural way.\footnote{Let $R \subseteq N$ be the set of computer players, with common-knowledge strategy profile $\sigma^*_R = (\sigma^*_n)_{n \in R}$. An assessment $(\sigma_N, (\overline{\sigma}^I)_{I \in \mathcal{I}_N})$ is a SCE if it is cursed-consistent, and is locally rational for the human players: for every human player $n\in  N\setminus R$ and every information set $I\in\mathcal{I}_n$, player $n$'s strategy $\sigma_n$ and their self-conjecture $\overline{\sigma}_n^I$ yield the same action distribution at $I$, and $\overline{\sigma}_n^I$ maximizes $n$'s expected utility when the play is initialized by $\rho(\overline{\sigma}_n^I,I)$ and proceeds according to $\overline{\sigma}_{-n}^I$. The computers play as prescribed $\sigma_R = \sigma^*_R$. WPCE extends similarly.}

\subsection{Pivotal voting}\label{sec:pivotal_voting}

\cite{esponda2014hypothetical} conduct a voting experiment to study hypothetical thinking. In their experiment, optimal play requires the subject to condition on their vote being pivotal. They find that most subjects fail to play optimally when they vote simultaneously, but play optimally when they vote after the other players.

The experiment proceeds as follows. A ball is selected at random from a jar with $10$ balls. The proportion of red balls is $p$, and the proportion of blue balls is $1-p$, for parameter $p\in\{.1,.2,\ldots,.9\}$. Three players vote, choosing either \textit{red} or \textit{blue}. If a majority vote for the correct color, the subject gets a payoff of $\$2$. Otherwise, the subject's payoff is $\$0$.

The three players consist of the lab subject and two computers. Both computers observe the color of the ball, and follow the same strategy: if the ball is red, vote \textit{red}; if the ball is blue, vote \textit{blue} with probability $q$ and \textit{red} with probability $1-q$, for parameter $q \in \{0.1,0.25,0.5,0.75,0.9\}$. The participant does not observe the color of the ball, but knows the computer strategies and the parameters $p$ and $q$.

The participant's vote only matters when it is pivotal---that is, when one computer votes \textit{red} and the other votes \textit{blue}. Conditional on that event, the ball is blue for sure, so the payoff-maximizing strategy is to vote \textit{blue} regardless of the parameters $p$ and $q$.

The experiment has two main treatments. In the \textit{simultaneous treatment}, all three players vote simultaneously. In the \textit{sequential treatment}, the subject observes the computers' votes before choosing their own vote. Each treatment proceeds for $45$ rounds without feedback, with varying parameters $p$ and $q$.

\cite{esponda2014hypothetical} label a subject `strategic' if they always vote \textit{blue} in the simultaneous treatment, and vote \textit{blue} whenever they are pivotal in the sequential treatment. They find that $22$ percent of subjects are strategic in the simultaneous treatment, compared to $76$ percent of subjects in the sequential treatment.

We say that an equilibrium has \textbf{strategic voting} if the subject votes \textit{blue} in the simultaneous treatment, and votes \textit{blue} whenever they are pivotal in the sequential treatment. We say that an equilibrium has \textbf{na\"ive voting} if the subject votes \textit{red} if $p > .5$ and \textit{blue} if $p < .5$ in the simultaneous treatment, and behaves in that way whenever they are pivotal in the sequential treatment.

Neither BNE nor CE explains the difference between the simultaneous treatment and sequential treatment. BNE predicts strategic voting in both treatments. In CE, the subject responds as though the computers' votes are independent of their private information, so CE predicts na\"ive voting in both treatments.\footnote{\cite{eyster2005cursed} define $\chi$-cursed equilibria, for parameter $\chi \in [0,1]$ that interpolates between BNE and CE. For any parameters $p$ and $q$, there is a unique threshold such that for all $\chi$ strictly above that threshold, $\chi$-CE predicts na\"ive voting in both treatments, and for $\chi$ strictly below that threshold, $\chi$-CE predicts strategic voting in both treatments.}

For the treatments in this experiment, SCE and WPCE make the same predictions.

WPCE predicts na\"ive voting in the simultaneous treatment. This treatment is a simultaneous Bayesian game, so by \Cref{thm:ICE_SCE} WPCE coincides with ICE.

WPCE predicts strategic voting in the sequential treatment. Consider the extensive game; at the initial history, nature plays $\omega_{\text{red}}$ with probability $p$ and $\omega_{\text{blue}}$ with probability $1-p$. Suppose that one computer votes \textit{red}, and the other computer votes \textit{blue}, and consider the corresponding information set of the subject. This information set is reached with positive probability, so the cursed-plausible conjecture is pinned down by the strategy profile and Bayes' rule. Conditional on that information set, nature has played $\omega_{\text{blue}}$ for sure. In the resulting conjecture, nature plays $\omega_{\text{blue}}$ for sure, one computer adopts the coarse strategy that plays \textit{blue} for sure, and the other computer adopts the coarse strategy that plays \textit{red} for sure. (Of course, the computers are not actually playing coarse strategies; but cursed-plausibility leads the subject to effectively infer the information content of being pivotal.) Local rationality implies that the subject votes \textit{blue}, resulting in strategic voting.

\subsection{Learning from prices}\label{sec:learning_from_prices}

\cite{ngangoue2021learning} conduct an experiment to study how traders learn from prices about a common-value risky asset. They compare a simultaneous treatment, in which the trader submits a limit order, to a sequential treatment, in which the trader observes the price before deciding whether to trade. Both treatments are equivalent in theory, but \cite{ngangoue2021learning} find that subjects make better inferences when observing a realized price than when thinking about a hypothetical price. 

The experiment studies a game with two players, trader $1$ and trader $2$. They can each buy or sell one unit of a common-value asset; its value $V$ is a random variable, with support on $\{\underline{v},\overline{v}\}$, for $\underline{v} < \overline{v}$, each with probability $.5$. Trader $n$'s payoff from buying at price $p_n$ is $V - p_n$ and their payoff from selling at price $p_n$ is $p_n - V$. Each trader has a type $t_n$, a random variable on $[0,1]$ with density
\begin{equation}
    f(t \mid V)=
        \begin{cases}
            2(1-t) & \text{if } V=\underline{v}\\
            2t & \text{if } V=\overline{v}
        \end{cases}.
\end{equation}
The private signals are independent conditional on $V$. Note that $E[V \mid t_n] = \underline{v} + t_n(\overline{v} - \underline{v})$. To ease notation, we normalize $\underline{v}$ to 0 and $\overline{v}$ to 1, so that $E[V \mid t_n] = t_n$.

In the \textbf{simultaneous learning treatment}, the order of play is as follows:
\begin{enumerate}
    \item Price $p_1$ is drawn $U[0, 1]$, independently of the common value and the signals.
    \item Trader $1$ observes their price $p_1$ and their type $t_1$ and either buys or sells.
    \item If trader $1$ buys, the price rises, and if trader $1$ sells, the price falls, according to the known formula
\begin{equation}\label{eq:price_update}
        p_2 = \begin{cases}
        \frac{1+p_1}{2} &  \text{ if $1$ bought}\\
        \frac{p_1}{2} &  \text{ if $1$ sold}
        \end{cases}.
\end{equation}
    \item Trader $2$ observes trader $1$'s price $p_1$ and their own type $t_2$, but does not observe trader $1$'s choice or the resulting price $p_2$.
    \item Trader $2$ chooses a limit order $b_2$, buying if $p_2 \leq b_2$ and selling otherwise.\footnote{\cite{ngangoue2021learning} also permit reverse limit orders, allowing the trader to sell at prices below $b_2$ and buy at prices above $b_2$. Such orders are sub-optimal in theory and rare in the data, so we omit them for ease of exposition.}
\end{enumerate}
The \textbf{sequential learning treatment} is the same, except for the last two steps:
\begin{itemize}
    \item [4*.] Trader $2$ observes trader $1$'s price $p_1$, their own type $t_2$, {and their own price $p_2$}.
    \item [5*.] Trader $2$ chooses to buy or sell at price $p_2$.
\end{itemize}

We now study the BNE of these games. In both treatments, trader $2$'s choices do not affect trader $1$'s payoffs. Upon observing $t_1$, trader $1$'s best response is to buy if $p_1 \leq t_1 = E[V \mid t_1]$ and to sell otherwise.\footnote{Trader $1$'s best response is generically unique, up to indifference at $p_1 = t_1$.}

The price updating rule \eqref{eq:price_update} reflects a no-arbitrage condition; it sets $p_2$ to be equal to the posterior expectation of the common value $\theta$, conditional on $p_1$ and trader $1$'s choice. Thus, if trader $2$ lacked a private signal, they would be indifferent between buying and selling at \textit{every} realization of $p_2$. Trader $2$'s private signal breaks that indifference. In every BNE of the simultaneous trading game, trader $2$ submits a limit order that buys for sure if $t_2 > .5$ and submits a limit order that sells for sure if $t_2 < .5$.\footnote{Such an order must satisfy $b_2 \in \left(\frac{1 + p_1}{2},1\right]$ if $t_2 > .5$, and $b_2 \in \left[0, \frac{p_1}{2}\right)$ if $t_2 < .5$.} In every BNE of the sequential trading game, trader $2$ buys if $t_2 > .5$ and sells if $t_2 < .5$.

In cursed equilibrium, trader $1$ behaves as in BNE, but trader $2$ ignores the information content of trader $1$'s actions, and hence ignores the information content of the price $p_2$. We have $E[V \mid t_2] = t_2$, so a cursed trader $2$ thus strictly prefers to buy if $t_2 > p_2$ and strictly prefers to sell if $t_2 < p_2$. In every CE of the simultaneous trading game, trader $2$ submits a limit order that ensures this happens; one such order is $b_2 = t_2$.\footnote{More generally, such a limit order satisfies $b_2 < \frac{p_1}{2}$ if $t_2 < \frac{p_1}{2}$, $b_2 \in \left(\frac{p_1}{2}, \frac{1 + p_1}{2}\right)$ if $t_2 \in \left(\frac{p_1}{2}, \frac{1 + p_1}{2}\right)$, and $b_2 > \frac{1 + p_1}{2}$ if $t_2 > \frac{1 + p_1}{2}$.} In every CE of the sequential trading game, trader $2$ buys if $t_2 > p_2$ and sells if $t_2 < p_2$. Let us call such trader $2$ behavior \textbf{na\"ive trading}.

\cite{ngangoue2021learning} adopt a conservative measure of na\"ive trading, looking only at realizations of $p_1$, $p_2$, and $t_2$ for which the na\"ive response differs from the empirical best-response. Based on this measure, they find that $37$ percent of such trades are consistent with na\"ive trading in the simultaneous treatment, compared to $19$ percent in the sequential treatment. That is, subjects are better at making inferences from realized prices, compared to making inferences from hypothetical prices.

These data are at odds with both the BNE prediction and the CE prediction. BNE does not predict na\"ive trading in either treatment, and CE predicts na\"ive trading in both treatments.

Since nature and trader $2$ have continuum action sets, we apply WPCE rather than SCE to the trading games.\footnote{If we discretize the action sets so that SCE is applicable, it would yield similar results to WPCE.} In any WPCE, trader $1$ behaves as in BNE, buying if $t_1 > p_1$ and selling if $t_1 < p_1$.

Consider the simultaneous treatment.  Trader $2$ observes $p_1$ and $t_2$. In every cursed-plausible conjecture at that information set, nature plays $\overline{v}$ with probability $t_2$, and trader $1$ plays a coarse strategy that is independent of $t_1$. Thus, WPCE coincides with CE in the simultaneous treatment, predicting na\"ive trading.

Consider the sequential treatment. Trader $2$ observes $p_1$, $p_2$, and $t_2$. In every cursed-plausible conjecture at that information set, nature plays $\overline{v}$ with probability equal to the Bayesian posterior probability, conditional on $p_1$, $p_2$, and $t_2$. Thus, WPCE coincides with BNE in the sequential treatment, predicting no na\"ive trading. 

WPCE formalizes the explanation offered by \cite{ngangoue2021learning}, who write,
\begin{quote}
    [\ldots] we find that the degree of ‘cursedness of beliefs’ is higher when the information contained in the price is less accessible: with price not yet realized, traders behave as if they tend to ignore the connection between other traders’ information and the price.
\end{quote}

\subsection{Common-value auctions with \textit{ex post} learning.}\label{sec:auctions_with_learning}

The experiments in \Cref{sec:pivotal_voting} and \Cref{sec:learning_from_prices} concerned games in which one player perfectly observes the other players' moves. But SCE makes novel predictions even when one player only partially observes the other players' moves. We now study such an application.

\cite{moser2019hypothetical} conducts an experiment in which subjects bid in a second-price auction for a common-value object, then learn whether their bid is higher than their opponent's bid, and then can revise their bid. Subjects substantially revise their bids in response to this information, contrary to the predictions of BNE and CE.\footnote{The experiment in \cite{moser2019hypothetical} does not have a treatment with feedback. Arguably, experiments without feedback are less apt for testing theories of equilibrium, such as those we consider here.} 

First, we describe the initial auction (stage $1$). There are two players $\{1,2\}$, each with type $t_n$ independently and uniformly distributed on the set $\mathbb{T} = \{0,1,\ldots , 10\} \cup \{50, 51, \ldots , 60\}$. The common value of the object is $V \equiv t_1 + t_2$.

Both players participate in a second-price auction for the object, simultaneously choosing integer bids $b_n$ between $0$ and $120$. The higher bidder wins, and pays the second-highest bid, for payoff $V - \min\{b_1,b_2\}$. Ties are broken in favor of the player with the higher type. If the types are equal, then both players have a payoff of $0$.

Considered by itself, stage $1$ has a unique symmetric BNE. (There are asymmetric BNE that involve weakly dominated bids.) In that BNE, player $n$ bids $b_n(t_n) = 2t_n$. Moreover, this BNE is an \textit{ex post} equilibrium; bidding $b_n(t_n) = 2t_n$ is a best-response to any realization of the other player's type.\footnote{The payoff to player $1$ from winning the object is $t_1 + t_2 - b_2(t_2) = t_1 - t_2$ and the payoff from losing is $0$. Hence, it is optimal to bid $b_1 = 2t_1$, which ensures that player $1$ wins whenever $t_1 > t_2$ and loses whenever $t_1 < t_2$.}

Considered by itself, stage $1$ also has a symmetric CE. A cursed player $n$ responds as though the other player's bid $b_{-n}$ is independent of their type $t_{-n}$. Given this conjecture, it is optimal to bid $b_n(t_n) = t_n + E[t_{-n} \mid t_n] = t_n + 30$. From the perspective of the cursed player, this bid is always a best-response.\footnote{That is, if we restrict the opponent to strategies that do not depend on their type, then the bid $t_n + 30$ does at least as well as any other bid.} In this setting, CE predicts higher bids than BNE for $t_n \in \{0,1,\ldots,10\}$, and lower bids than BNE for $t_n \in \{50,51,\ldots,60\}$.

We now describe stage $2$. Player $1$ is informed whether $b_1 \leq b_2$ or $b_1 > b_2$, and then can revise their bid to $\hat{b}_1$. Player $1$ is compensated for stage $1$ and stage $2$ separately. Player $2$ is only compensated for the results of stage $1$.

The symmetric stage $1$ bidding functions that we described above remain equilibrium behavior even when stage $2$ is appended.\footnote{In the experiment of \cite{moser2019hypothetical}, subjects learn the details of stage $2$ only after stage $1$ is completed, but it makes no difference to the BNE prediction and the CE prediction.}

Given symmetric BNE bidding in stage $1$, repeating player $1$'s original bid in stage $2$ is a best response regardless of whether $b_1 \leq b_2$ or $b_1 > b_2$. Thus, the two-stage game has a BNE in which the stage $2$ bidding function is the same as in stage $1$, that is $\hat{b}_1(t_1) = 2t_1 = b_1(t_1)$.

If $b_2$ is independent of $t_2$, then it is optimal for player $1$ to bid $t_1 + 30$, regardless of the realization of $b_2$. Thus, the two-stage game has a CE in which the stage $2$ bidding function is the same as in stage $1$, that is $\hat{b}_1(t_1) = t_1 + 30 = b_1(t_1)$.\footnote{Here we apply CE by reducing the game to its Bayesian normal form. For each type $t_1$, player $1$ chooses a stage $1$ bid, a stage $2$ bid conditional on learning that $b_1 \leq b_2$, and a stage $2$ bid conditional on learning $b_1 > b_2$.}

In this sense, both BNE and CE predict that the new information does not cause any change in player $1$'s bids. However, \cite{moser2019hypothetical} reports that subjects react substantially to the information, raising their bids when informed that $b_1 \leq b_2$ and lowering their bids when informed that $b_1 > b_2$. Upon learning that $b_1 \leq b_2$, $79\%$ to $83\%$ of bids are revised upwards, and upon learning that $b_1 > b_2$, $42\%$ to $90\%$ of bids are revised downwards.\footnote{The range reflects variation that depends on player $1$'s type $t_1$. For a detailed breakdown, see Figure 6 of \cite{moser2019hypothetical}.}

SCE provides an explanation for this pattern of behavior. In particular, at any stage $1$ information set, both players' SCE conjectures are the same as in cursed equilibrium, so there exists a SCE in which both players bid $b_n(t_n) = t_n + 30$ in stage $1$, with player $1$ expecting to repeat that same bid in stage $2$. Take any on-path stage-$2$ information set at which player $1$ with type $t_1 = t$ learns that $b_1 \leq b_2$. Then cursed-consistency implies that player $1$ conjectures that $t_2$ is independent of player $2$'s bid and uniformly distributed on $\{t' \in \mathbb{T}: t \leq t'\}$. (This is different from the actual distribution; $2$'s bid is exactly $t_2 + 30$, for $t_2$ distributed as above.) It is a locally rational for player $1$ to revise their bid to 
\begin{equation}\label{eq:SCE_revised_bid}
    E[t_1 + t_2 \mid t_1 = t \leq t_2],
\end{equation}
if necessary rounded to one of the two nearest integers.\footnote{Recall that bids are restricted to integers between $0$ and $120$.} Note that \eqref{eq:SCE_revised_bid} weakly exceeds the original bid $b_n(t) = E[t_1 + t_2 \mid t_1 = t]$, with strict inequality for $t > 0$. Thus, player $1$ revises their bid upwards upon learning that $b_1 \leq b_2$. A symmetric argument establishes that player $1$ revises their bid downwards upon learning that $b_1 > b_2$.

\section{Continuous-time auctions}\label{sec:dynamic_auctions}

Cursed-plausibility formalizes the idea that players condition on all of their information, and not just their type. This idea has implications for common-value auctions, the primary motivating application of \cite{eyster2005cursed}. In particular, it might explain the different bidding behavior in dynamic auctions compared to static auctions, documented by \cite{levin1996revenue} and \cite{levin2016separating}.

WPCE and SCE are extensive-form solution concepts, whereas the standard models of dynamic auctions use continuous time. Discretized versions of those models are intractable except for special cases. To ease the comparison with standard results, we define solution concepts for Dutch and English auctions that formalize the idea that each bidder understands the marginal distribution of the other bidders' signals, conditional on all the information revealed so far, but neglects the link between the other bidders' signals and their future bids.

We use the standard model of common values with affiliated signals \citep{milgrom1982theory}. There is one object and a set of bidders $N$. Let $X_N = (X_1,\ldots,X_{|N|})$ be a vector of real-valued random variables (`signals'), one for each bidder, each with support on the interval $[\underline{x}, \overline{x}]$. Let $f(x_N)$ denote the joint probability density of the signals. We assume that $f$ is symmetric. We assume that the signals are affiliated, that is for all $x_N, x_N' \in [\underline{x},\overline{x}]^N$, we have
\begin{equation}
    f(x_N \vee x_N') f(x_N \wedge x_N') \geq f(x_N) f(x_N').
\end{equation}

The common value of the object is a random variable $V$, and the bidders have quasilinear utility, with payoff $V - p$ if they win the object at price $p$ and $0$ otherwise.  We denote
\begin{equation}
    w(x_1,\ldots,x_{|N|}) \equiv \E\left[V \mid X_N = (x_1,\ldots,x_{|N|})\right],
\end{equation}
and assume that the function $w$ is symmetric, non-negative, continuous, and non-decreasing. We assume that $\E[V] < \infty$. We also define
\begin{equation}
    {v}(x) \equiv \E\left[V \mid X_1 = x \right].
\end{equation}

Next we define $Y_1 \equiv \max_{n \neq 1} X_n$. Let $f_{Y_1}(y \mid x)$ be the density of the random variable $Y_1$ conditional on $X_1 = x$, and $F_{Y_1}(y \mid x)$ the corresponding cumulative distribution.

As usual, we focus on symmetric pure-strategy equilibria.
\begin{theorem}\label{thm:first_price}
    It is a cursed equilibrium of the first-price auction for each bidder $n$ with signal $X_n = x$ to bid $b^{\mathrm{1P}}(x)$, defined by the linear differential equation    \begin{equation}\label{eq:first_price_ODE}
        \frac{d}{d x}b^{\mathrm{1P}}(x) = [{v}(x) - b^{\mathrm{1P}}(x)]\frac{f_{Y_1}(x \mid x)}{F_{Y_1}(x \mid x)},
    \end{equation}    \begin{equation}\label{eq:first_price_boundary}
        b^{\mathrm{1P}}(\underline{x}) = {v}(\underline{x}).
    \end{equation}
\end{theorem}

In other words, in the first-price auction with common values, each bidder behaves as if they have a private value equal to $v(x)$, which is the expected common value conditional on their own signal.

In a Dutch auction, the price starts high (above $w(\overline{x},\ldots,\overline{x})$) and descends continuously until some bidder claims the object for that price, which ends the auction. The Dutch auction is strategically equivalent to the first-price auction, so cursed equilibrium makes the same prediction for both formats.

We now define a solution concept that modifies cursed equilibrium, to capture the idea that bidders in a Dutch auction understand the information content of realized events. Suppose that all bidders bid according to increasing function ${b}$. When the clock price is $p$, bidder $1$ infers that $Y_1 \leq {b}^{-1}(p)$. Consequently, conditional on $X_1 = x$ and facing a price of ${b}(y)$, bidder $1$ believes that upon winning at a price $\beta \leq b(y)$, the expected common value of the object is
\begin{equation}\label{eq:CCE_dutch}
    \overline{v}(x,y) \equiv \E\left[ V \mid X_1 = x, Y_1 \leq y \right].
\end{equation}
By contrast, the correct calculation would yield
\begin{equation}\label{eq:rational_dutch}
    \E\left[ V \mid X_1 = x, Y_1 \leq b^{-1}(\beta) \right].
\end{equation}
At price $\beta = b(y)$, Expression \eqref{eq:CCE_dutch} coincides with the correct calculation \eqref{eq:rational_dutch}, but contingent on the price falling further to $\beta < b(y)$, Expression \eqref{eq:CCE_dutch} overestimates the common value.

Thus, we define a (symmetric) \textbf{clock cursed equilibrium of the Dutch auction} as a bidding function ${b}$ such that ${b}(x)$ is in
\begin{equation}\label{eq:CCE_Dutch_obj}
    \argmax_\beta \int_{\underline{x}}^{{b}^{-1}(\beta)}(\overline{v}(x,x) - \beta) f_{Y_1}(\alpha \mid x) d\alpha .
\end{equation}
Expression \eqref{eq:CCE_Dutch_obj} implies that the bidder with signal $X_n = x$ is willing to claim the object when the price hits ${b}(x)$. We additionally require that the bidder is willing to wait at all prices that exceed ${b}(x)$, that is, for all $x' > x$, there exists $\beta < b(x')$ that maximizes
\begin{equation} \label{eq:CCE_Dutch_wait}   \int_{\underline{x}}^{{b}^{-1}(\beta)}(\overline{v}(x,x') - \beta) f_{Y_1}(\alpha \mid x) d\alpha.
\end{equation}

\begin{theorem}\label{thm:dutch} Assume that $\overline{v}(x,y)$ is increasing in $x$. It is a clock cursed equilibrium of the Dutch auction for each bidder $n$ with signal $X_n$ to bid $b^{\mathrm{Dutch}}(X_n)$, defined by the linear differential equation
    \begin{equation}\label{eq:dutch_ODE}
        \frac{d}{d x}b^{\mathrm{Dutch}}(x) = [\overline{v}(x,x) - b^{\mathrm{Dutch}}(x)]\frac{f_{Y_1}(x \mid x)}{F_{Y_1}(x \mid x)},
    \end{equation}
    \begin{equation}\label{eq:dutch_boundary}
        b^{\mathrm{Dutch}}(\underline{x}) = \overline{v}(\underline{x},\underline{x}).
    \end{equation}
\end{theorem}

The characterization of $b^{\mathrm{Dutch}}$ differs from $b^{\mathrm{1P}}$ only by replacing $v(x)$ with $\overline{v}(x,x)$. The expectation $\overline{v}(x,x)$ conditions not only on $X_1 = x$, but also on $Y_1 \leq x$; by affiliation, it follows that $\overline{v}(x,x) \leq v(x)$. Thus, the clock cursed equilibrium of the Dutch auction has lower bids than the cursed equilibrium of the first-price auction.
\begin{theorem}\label{thm:dutch_vs_1P}
    $b^{\mathrm{Dutch}}(x) \leq b^{\mathrm{1P}}(x)$ for all $x$.
\end{theorem}

In a Dutch auction, the descending price reveals at each point an upper bound on the other bidders' signals. Under clock cursed equilibrium, this realized uncertainty causes the bidder to update negatively about the common value, lowering bids relative to the first-price auction. This is consistent with the results of the experiment reported in \cite{levin2016separating}.

Next, we describe behavior in the second-price auction.
\begin{theorem}[\citet{eyster2005cursed}]
    It is a cursed equilibrium of the second-price auction for each bidder $n$ with signal $X_n = x$ to bid $b^{\mathrm{2P}}(x) = {v}(x)$.
\end{theorem}

In other words, in a second-price auction with common values, each bidder places a bid equal to the expected common value conditional on their own signal.

We compare the second-price auction to the {silent English auction}. In this auction, the price rises continuously from $0$ and bidders make irrevocable decisions to quit. Each bidder does not observe when other bidders quit. When only one bidder is left, the auction ends and the remaining bidder wins at the going price.

This is a plausible model of the format used at major auction houses, such as Christie's and Sotheby's. In such auctions, most bidders do not bid audibly, but instead use secret gestures to communicate with the auctioneer.\footnote{These gestures “may be in the form of a wink, a nod, scratching an ear, lifting a pencil, tugging the coat of the auctioneer, or even staring into the auctioneer’s eyes---all of them perfectly legal” \citep{cassady1967}.}

Silent English auctions are a useful benchmark, because they are strategically equivalent to second-price auctions. In particular, cursed equilibrium makes the same predictions for both formats.

Suppose that all bidders in a silent English auction choose quitting points according to increasing function ${b}$. Essentially, we will require that the bidder correctly conditions on observed information, that is, at clock price $b(y)$, the bidder knows that $X_1 = x$ and $Y_1 \geq y$.  But the bidder neglects hypothetical events, treating the common value conditional on winning at price $\beta \geq b(y)$ as
\begin{equation}\label{eq:CCE_silent}
    \underline{v}(x,y) \equiv E[V \mid X_1 = x, Y_1 \geq y],
\end{equation}
whereas the correct calculation yields 
\begin{equation}\label{eq:rational_silent}
    E[V \mid X_1 = x, Y_1 = {b}^{-1}(\beta)].
\end{equation}

A \textbf{clock cursed equilibrium of the silent English auction} is a bidding function ${b}$ such that for all $x' < x$, there exists $\beta > b(x')$ that maximizes
\begin{equation}\label{eq:CCE_english_wait}
    \E \left[ (\underline{v}(x,x') - {b}(Y_1))1_{{b}(Y_1) \leq \beta} \mid X_1 = x, Y_1 \geq x' \right],
\end{equation}
and ${b}(x)$ is in
\begin{equation}\label{eq:CCE_english_obj}
    \argmax_{\beta} \E \left[ (\underline{v}(x,x) - {b}(Y_1))1_{{b}(Y_1) \leq \beta} \mid X_1 = x, Y_1 \geq x \right].
\end{equation}


Expression \eqref{eq:CCE_english_wait} requires that the bidder with signal $X_1 = x$ is willing to wait at all prices below $b(x)$, while expression \eqref{eq:CCE_english_obj} requires that the bidder quits at price $b(x)$.

\begin{theorem}\label{thm:silent_english} Assume that $ \underline{v}(x,y)$ is increasing in $x$. It is a clock cursed equilibrium of the silent English auction for each bidder $n$ with signal $X_n = x$ to quit at price $b^{\mathrm{silent}}(x) = \underline{v}(x,x)$. 
\end{theorem}
\begin{proof}
    Observe that for ${b} = b^{\mathrm{silent}}$, \eqref{eq:CCE_english_wait} is maximized at $\beta = \underline{v}(x,x') > \underline{v}(x',x') = {b}(x')$. Moreover, \eqref{eq:CCE_english_obj} is maximized at $\beta = \underline{v}(x,x) = {b}(x)$.
\end{proof}

The expectation $\underline{v}(x,x)$ conditions not only on $X_1 = x$ but also on $Y_1 \geq x$; by affiliation, it follows that $\underline{v}(x,x) \geq v(x)$. Thus, in silent English auctions, the rising price reveals a lower bound on $Y_1$, and that raises bids compared to second-price auctions.

\begin{theorem}\label{thm:silent_vs_2P}
    $b^{\mathrm{silent}}(x) \geq b^{\mathrm{2P}}(x)$ for all $x$.
\end{theorem}

The canonical model of English auctions assumes that each bidder observes when other bidders quit. \cite{levin1996revenue} find in a lab experiment that ``bidders use the ``public information" inherent in other bidders' drop-out prices to (largely) overcome the winner's curse in English auctions." We now formalize this idea.

A bidding strategy in a canonical English auction is a profile of functions $b_k(x \mid p_1,\ldots,p_k)$ for $k \in \{0,\ldots, |N|-2\}$, that specifies the quit point of a bidder with signal $X_n = x$ when $k$ bidders have quit at prices $p_1,\ldots,p_k$. We use $Y_{k}$ to denote the $k$th highest signal among bidder $1$'s opponents. As usual, we restrict attention to symmetric pure strategy equilibria that are increasing in each bidder's own signal, so that bidder $1$ can, upon observing the $k$th quit price $p_k$, infer that the corresponding signal is $b^{-1}_{k-1}(p_k \mid p_1, \ldots, p_{k-1})$. Thus, we abuse notation to write the bidding function as depending on the revealed signals $b_k(x \mid y_{|N|-1},\ldots,y_{|N|-k})$.

Suppose that bidder $1$ has signal $X_1 = x$, has observed $k$ quit prices with corresponding signal realizations $y_{|N|-1},\ldots,y_{|N|-k}$, and the current price is $b_k(x'\mid y_{|N|-1},\ldots,y_{|N|-k})$. We posit that bidder $1$ believes that the common value conditional on winning at price $\beta$ is
\begin{equation}
\begin{split}
    &v_k(x,y_{|N|-1},\ldots,y_{|N|-k},x') \\ & \equiv \E [V \mid X_1 = x, Y_{|N|-1} = y_{|N|-1},\ldots, Y_{|N|-k} = y_{|N|-k}, Y_{|N|-k-1} \geq x' ] .
\end{split}   
\end{equation}
Clock cursed equilibrium for the canonical English auction is defined analogously as for the silent English auction. A formal definition is in \Cref{sec:CCE_canon_def}.
\begin{theorem}\label{thm:canon_english}
    Assume that $v_k(x,y_{|N|-1},\ldots,y_{|N|-k},x')$ is increasing in $x$ for all $k$ and all $x \geq x' \geq y_{|N|-k}$. It is a clock cursed equilibrium of the canonical English auction for each bidder $n$ with signal $X_n = x$ who has observed quit points corresponding to signals $y_{|N|-1},\ldots,y_{|N|-k}$ to quit at price $b_k^{\mathrm{canon}}(x \mid y_{|N|-1},\ldots,y_{|N|-k}) = v_k(x,y_{|N|-1},\ldots,y_{|N|-k},x)$. 
\end{theorem}
\begin{proof}
    This follows the same steps as the proof of \Cref{thm:silent_english}.
\end{proof}

One implication of \Cref{thm:canon_english} is that the winner's curse vanishes under certain conditions, in the sense that the losses of the winning bidder are small when there are many bidders.  In particular, suppose we have an infinite sequence of bidders, each with signal $X_n$, and consider the sequence of auctions involving just the first $K$ bidders, each with corresponding functions $w^K$ and $v_k^K$. Let $\overline{X}_k^K$ denote the $k$th highest of the first $K$ signals. For the equilibrium described in \Cref{thm:canon_english}, the clearing price in the auction with $K$ bidders is
\begin{equation}\label{eq:clearing_price}
    v^K_{K-2}\left(\overline{X}_{2}^{K}, \overline{X}_{K}^{K},\ldots, \overline{X}_{3}^{K},\overline{X}_{2}^{K} \right).
\end{equation}
Suppose that the posterior on $V$ does not depend too much on the highest signal for large $K$; more formally, that \eqref{eq:clearing_price} converges in mean to 
\begin{equation}
    w^K\left(\overline{X}_{2}^{K}, \overline{X}_{K}^{K},\ldots, \overline{X}_{3}^{K},\overline{X}_{1}^{K} \right) = w^K\left(X_1,X_2,\ldots, X_K \right).\footnote{The equality is by symmetry of $w^K$. This convergence holds, for instance, for the signal structure studied by \cite{levin1996revenue}, and also for $w^K(X_1,X_2,\ldots,X_K) = \sum_{k=1}^K X_k / K$ with iid signals.}
\end{equation}
Then the winner's expected payoff under the clock cursed equilibrium converges to $0$. By contrast, \cite{eyster2005cursed} find that the losses under cursed equilibrium can be large even when there are many bidders.\footnote{Fully cursed equilibrium makes the same predictions in canonical English auctions and second-price auctions, so this follows from Proposition 6 of \cite{eyster2005cursed}.}

In summary, we studied a variation of cursed equilibrium for clock auctions, positing that each bidder's belief about the other bidders' signals conditions not just on their own signal, but on all the information at hand. This predicts that bids are reduced in Dutch auctions, compared to first-price, as reported by \cite{levin2016separating}. It also predicts that the winner's curse is small in canonical English auctions with many bidders, as reported by \cite{levin1996revenue}. One novel (and untested) prediction is that silent English auctions lead to higher bids than second-price auctions.\footnote{One can define a partially-cursed generalization of clock cursed equilibrium, that parallels \Cref{sec:chi_psi_SCE}. Such a solution concept would specify that each bidder assesses the common value with weight $\chi$ on the expected value conditional on all the realized information, and weight $1-\chi$ on the correct calculation.}

\section{Discussion}\label{sec:discussion}

We close by discussing some limitations of sequential cursed equilibrium.

The predictions of sequential cursed equilibrium depend on the action labels. Those labels affect both the coarsest valid partition $\mathcal{F}$, and the mapping from strategy profiles to coarse strategies. For instance, we can take the simultaneous trading game of \Cref{sec:introduction}, and define two new actions for player $1$, action $ad$ that accepts at type $\{\omega_1,\omega_2\}$ and declines at type $\{\omega_3\}$, and action $da$ that declines at type $\{\omega_1,\omega_2\}$ and accepts at type $\{\omega_3\}$. Player $1$'s dominant strategy is to play action $ad$ at both types, which changes the SCE prediction. This suggests that SCE is too permissive if the modeler can arbitrarily specify the action labels. 

Action labels matter even for Nash equilibrium. For example, consider matching pennies in the extensive form.  One player moves second, and their action labels imply that they cannot cause \textit{heads} if the first player chose \textit{heads} and yet cause \textit{tails} if the first player chose \textit{tails}. Swapping the labels at one history would change the Nash equilibrium prediction. 

Action labels are a modeling choice. Given what they know, the player can enact only some mappings from states of the world to consequences. When we translate the world to a model, we use the action labels to express which mappings are feasible.

In SCE, the action labels matter even across information sets. When we put two information sets in the same coarse set, we posit a counterfactual in which the player chooses without knowing which situation has obtained. For example, consider a poker player who bets without seeing their hand, or a bidder who bids without seeing their estimate of the object's value.\footnote{Such counterfactuals are not always coherent; for instance, the player presumably cannot be ignorant about whether they are choosing from two options or from three options.} In our view, the action labels should reflect what mappings are feasible in that counterfactual.

Motivated by the preceding discussion, we propose the following modeling convention: Imagine a person called to play in two situations $I$ and $I'$, who somehow does not know whether they are in situation $I$ or in situation $I'$. If anything they could do ({e.g.}\ any bet they could place, any ballot they could cast) that causes consequence $a$ at situation $I$ would also cause consequence $a'$ at situation $I'$, and \textit{vice versa}, then $a$ and $a'$ should have the same action label. This convention is not a formal resolution to the labeling critique, but it helps to tie the modeler's hands.

Another limitation of SCE is that player $1$ neglects the link between player $2$'s actions and player $2$'s information, even when $2$'s information solely concerns what player $1$ has done. For instance, in the game depicted in \Cref{fig:causal_illus}, the union of player 2's information sets is a coarse set, so there is a SCE in which player $1$ plays $L$, conjecturing that $2$ will play $l$ for sure regardless of player $1$'s action.
\begin{figure}[h]
    \centering
\begin{tikzpicture}[font=\footnotesize,edge from parent/.style={draw,thick}]
    \tikzstyle{solid}=[circle,draw,inner sep=1.2,fill=black];
    \tikzstyle{hollow}=[circle,draw,inner sep=1.2]
    \tikzstyle{level 1}=[level distance=5mm,sibling distance=30mm]
    \tikzstyle{level 2}=[level distance=12mm,sibling distance=12mm]
    \node(0)[hollow]{}
    child{node[solid]{}
        child{node{$1,1$}
        edge from parent node[left]{$l$}
        }
        child{node{$0,0$}
        edge from parent node[right]{$r$}
        }
    edge from parent node[above]{$L$}
    }
    child{node[solid]{}
        child{node{$0,0$}
        edge from parent node[left]{$l$}
        }
        child{node{$3,3$}
        edge from parent node[right]{$r$}
        }
    edge from parent node[above]{$R$}
    }
     ;

      \node[above,yshift=2]at(0){$1$};
      \node[above,yshift=2]at(0-1){$2$};
      \node[above,yshift=2]at(0-2){$2$};
\end{tikzpicture}
    \captionsetup{width=.68\textwidth}
    \caption{}
    \label{fig:causal_illus}
\end{figure}

One might find it plausible that each player should understand how their own actions affect other players' actions. SCE can be modified to allow this; for \textbf{causal sequential cursed equilibrium} (causal SCE), instead of having one conjecture at each information set, we have one conjecture for each \textit{action} at that information set. When player $1$ considers playing action $a$ at information set $I$, and thinks about player $2$'s play at coarse set $F$, player $1$ correctly understands the counterfactual distribution of $2$'s actions at $F$, were player $1$ to play $a$ at $I$. The game in \Cref{fig:causal_illus} has a unique causal SCE, in which $1$ plays $R$, and $2$ plays $l$ after $L$ and $r$ after $R$. A full definition of causal SCE is in \Cref{app:causal_SCE}.

Finally, a third limitation of SCE is that it does not distinguish between subtle hypothetical events and salient hypothetical events. \cite{esponda2021contingent} reframe several standard decision problems to focus on the contingencies at which subjects' choices can affect payoffs. They find that this causes subjects to behave more in line with the rational model. Framing effects such as these are not captured by SCE, since its predictions depend only on the extensive game.

\bibliographystyle{ecta}
\bibliography{references}

\appendix

\section{Proofs omitted from the main text}\label{app:proofs}

\subsection{Proof of Theorem \ref{thm:coarsest_valid}}\label{app:proof_coarsest_valid}

We start by constructing the valid partition $\mathcal{F}$ of $H \setminus Z$ and later show that $\mathcal{F}$ is coarser than any other valid $\mathcal{J}$.

First, we define a partial order on histories at which player $n$ is called to play, the set $H_n \equiv \cup \mathcal{I}_n$. Let the $n$-predecessor of a history be the partial history that ends at the last time $n$ is called to play. Formally, given $h\in H_n$, the $n$-predecessor of $h$, denoted $\pi_n(h)$, is $h' \in H_n$ such that $h' \prec h$ and there does not exist $h'' \in H_n$ such that $h' \prec h'' \prec h$. If no $n$-predecessor of $h$ exists, we define $\pi_n(h)\equiv\emptyset$.

We state the usual definition of perfect recall \cite[p.\ 224]{mas1995microeconomic}. A game has \textit{perfect recall} if players remember both their past information and their past actions. That is, for a game $G$, information set $I\in\mathcal{I}_n$ and histories $h,h'\in I$, 
\begin{enumerate}
    \item Either: $\pi_n(h) = \emptyset = \pi_n(h')$
    \item Or: $\pi_n(h) \neq \emptyset $ and $\pi_n(h') \neq \emptyset$, and
    \begin{enumerate}
        \item there exists $I'$ such that $\pi_n(h),\pi_n(h') \in I'$,
        \item and there exists an action $a\in A$, such that $(\pi_n(h),a)\preceq h$ and $(\pi_n(h'),a)\preceq h'$.
    \end{enumerate}
\end{enumerate}

For each player $n \in N$, we use induction to define the partition of $\mathcal{F}_n$ of $H_n$. Note that the partition of information sets for one player does not impact the validity of the information sets of the other players, hence we can manage the partition for each player separately. For player $n$, we use the following induction:
Start with the histories at which $n$ is called to play for the first time, $H_n^0\equiv\{h\in H_n : \pi_n(h)=\emptyset \}$. Let $\mathcal{F}_n^0$ be the partition of $H_n^0$ such that for each $h\in H_n^0$ with available actions $A(h)$, $\{h'\in H_n^0 : A(h')=A(h)\}\in \mathcal{F}^0_n$.

Inductively, let $H_n^{l}\equiv\{h\in H_n : \pi_n(h)\in H_n^{l-1}\}$ for $l \in \{1,2,\ldots\}$. We now define a partition $\mathcal{F}_n^l$ of $\cup_{k \leq l} H_n^k$. First, we require that for every $F \in \mathcal{F}_n^{l-1}$, we have $F \in \mathcal{F}_n^{l}$. Next, we specify $\mathcal{F}_n^l$ on the set $H_n^l$. In order for this partition to specify information sets for an extensive game of perfect recall, these three properties are necessary and sufficient:
\begin{enumerate}
    \item History $h \in H_n^l$ can only share a cell of the partition $\mathcal{F}_n^{l}$ with histories $h' \in H_n^l$ that have the same available actions.
    \item By perfect recall, the predecessors of $h$ and $h'$ must share a cell of the partition $\mathcal{F}_n^{l-1}$.
    \item By perfect recall, the same action must have been taken at the predecessors of $h$ and $h'$.
\end{enumerate}
Formally, this means that for each $h \in H_n^l$, we have
\begin{multline}
    \{ h'\in H_n^{l} : A(h)=A(h') \text{ and } \exists F\in \mathcal{F}^{l-1}_n:\{\pi_n(h),\pi_n(h') \}\subseteq F \\ \text{ and } \exists a\in A: (\pi_n(h'),a)\preceq h',(\pi_n(h),a)\preceq h \}\in \mathcal{F}_n.
\end{multline}
Since the game is well-behaved, this inductive procedure defines a partition of $\mathcal{F}_n$.

We define $\mathcal{F}$ to be the partition of $H$ that agrees with the partition $\mathcal{F}_n$ of $H_n$ for each $n \in N$, with singleton cells for the set of nature histories $H_\nature$.

Let $\mathcal{J}$ be a valid partition. We will prove that for all $h, h' \in H_n$, if $h$ and $h'$ are in the same cell of $\mathcal{J}$, then they are in the same cell of $\mathcal{F}$. Observe that since $\mathcal{J}$ is valid, $h$ and $h'$ are in the same cell of $\mathcal{J}$, and the game is well-behaved, it follows that $h, h' \in H_n^l$ for some player $n \in N$ and some integer $l$---otherwise, player $n$ has forgotten how many times they were called to play.

We argue by induction. Consider the claim: For all $h,h' \in H_n^k$, if $h$ and $h'$ are in the same cell of $\mathcal{J}$, then $h$ and $h'$ are in the same cell of $\mathcal{F}$.

Clearly the claim holds for $k = 0$. Suppose that the claim holds for all $k' < k$. We now prove that it holds for $k$. Take $h, h' \in H_n^{k}$ that are in the same cell of $\mathcal{J}$. Note that $\pi_n(h)$ and $\pi_n(h')$ are members of the same cell of $\mathcal{J}$ (by perfect recall), and hence (by the inductive claim) $\pi_n(h)$ and $\pi_n(h')$ are members of the same cell of $\mathcal{F}$.  By perfect recall, for some action $a\in A$, we have $(\pi_n(h),a)\preceq h$ and $(\pi_n(h'),a)\preceq h'$. In addition, $A(h)=A(h')$, otherwise $\mathcal{J}$ is not valid. Thus, there exists $F \in \mathcal{F}_n^k$ such that $h,h' \in F$, so $h$ and $h'$ are in the same cell of $\mathcal{F}$. Thus the claim holds for every integer $k$. The game is well-behaved, so $\bigcup_k H_n^k=H_n$, which completes the proof.

\subsection{Proof of Theorem \ref{thm:cursed_consistent_implies_plausible}}

Let $(\sigma_N^k, (\overline{\sigma}^{I,k})_{I \in \mathcal{I}_N})$ be a sequence of cursed-plausible assessments that converge to $(\sigma_N, (\overline{\sigma}^I)_{I \in \mathcal{I}_N})$, and define $\sigma^k \equiv (\sigma_N^k, \sigma_\nature)$ and $\sigma \equiv (\sigma_N,\sigma_\nature)$.

To show that the assessment $(\sigma_N, (\overline{\sigma}^I)_{I \in \mathcal{I}_N})$ is cursed-plausible, we need to show that for any player $n\in N$ and information set $I\in\mathcal{I}_n$, $\overline{\sigma}^I$ is cursed-plausible with $\sigma_N$. Take any such player and information set. The first and third requirements of cursed-plausibility follow by inspection. We now prove the second requirement holds.

Take any $m \neq n$, $I' \in \mathcal{I}_m$ such that $\measure[\sigma](I \cap F(I'))>0$. For any  $a \in A(I')$, we have that the right-hand side of \eqref{eq:CPS_update} is
    \begin{equation}
    \begin{split}
       & \measure[\sigma](\{(h,a): h \in F(I')\} \mid I \cap F(I'))  \\
       =  \lim_{k \rightarrow \infty} & \measure[\sigma^k](\{(h,a): h \in F(I')\} \mid I \cap F(I')) \\ =  \lim_{k\rightarrow\infty} & \overline{\sigma}_m^{I,k}(a,I') =  \overline{\sigma}_m^I(a,I'). 
    \end{split}
    \end{equation}   
The first equality follows by $\lim_{k \rightarrow \infty} \sigma^k = \sigma$ and $\measure[\sigma](I \cap F(I'))>0$. The second equality follows by definition of cursed-plausibility. The third equality follows by $\lim_{k \rightarrow \infty} \overline{\sigma}^{I,k} = \overline{\sigma}^I$. Thus, $(\overline{\sigma}^I)_{I \in \mathcal{I}_N}$ is cursed-plausible with $\sigma_N$.

\subsection{Proof of Theorem \ref{thm:SCE_existence}}

The proof proceeds in three steps.  In step I, we constrain players to play each action with probability at least $\epsilon > 0$, and define a best-response correspondence for the induced cursed-plausible conjectures, denoted $B^\epsilon$.  In step II, we show that $B^\epsilon$ has a fixed point.  In step III, we use a sequence of fixed points of $B^\epsilon$, taking $\epsilon$ to 0, to construct a sequential cursed equilibrium.

\subsubsection{Step I}

Let us fix some finite extensive game $G$. Given parameter $\epsilon > 0$, let $\Sigma_N^\epsilon$ be the set of strategy profiles that play each action with probability at least $\epsilon$, that is, for all players $n \in N$, all information sets $I \in \mathcal{I}_n$, and all actions $a \in A(I)$, we have $\sigma_n(a,I) \geq \epsilon$.     

Given some $\sigma_N \in \Sigma_N^\epsilon$, some player $n$, and some information set $I \in \mathcal{I}_n$, let $\overline{\sigma}^I$ be a conjecture that is cursed-plausible with $\sigma_N$. Since $\sigma_N$ is fully mixed, $\overline{\sigma}^I$ is uniquely pinned down at all compatible information sets except for $\{I' \in \mathcal{I}_n: I \prec I' \}$.  A distribution $d \in \Delta(A(I))$ is a \textbf{cursed $\epsilon$-best response} to $\sigma_N$ if there exists strategy $\tilde{\sigma}_n \in \Sigma_n^\epsilon$ such that $\tilde{\sigma}_n(I) = d$ and $\tilde{\sigma}_n$ maximizes $n$'s expected utility when play is initialized by $\rho(\overline{\sigma}^I,I)$ and proceeds according to $(\tilde{\sigma}_n,\overline{\sigma}^I_{-n})$, that is, if
\begin{equation}\label{eq:cursed_BR}
    \tilde{\sigma}_n \in \argmax_{\hat{\sigma}_n \in \Sigma_n^\epsilon} \left\{ \sum_{h \in I}  \rho(h, \overline{\sigma}^I,I) u_n(h, \hat{\sigma}_n, \overline{\sigma}_{-n}^I) \right\}.
\end{equation}
Since the cursed-plausible conjecture $\overline{\sigma}^I$ is unique up to behavior at $\{I' \in \mathcal{I}_n: I \prec I' \}$, it follows that the right-hand side of \eqref{eq:cursed_BR} is the same for all conjectures that are cursed-plausible with $\sigma_N$.
Similarly, we say that a strategy profile $\sigma'_N$ is a cursed $\epsilon$-best response to $\sigma_N$ if for all $n \in N$ and all $I \in \mathcal{I}_n$, the distribution $\sigma'_n(I)$ is a cursed $\epsilon$-best response to $\sigma_N$. 

Let $B^\epsilon(\sigma_N)$ denote the set of strategy profiles that are cursed $\epsilon$-best responses to $\sigma_N$, which defines a correspondence $B^\epsilon: \Sigma_N^\epsilon \rightrightarrows \Sigma_N^\epsilon$.

\subsubsection{Step II}
We fix $\epsilon > 0$ and show that $B^\epsilon$ has a fixed point by Kakutani's theorem. We now check the conditions of the theorem.

\begin{enumerate}
    \item $\Sigma_N^\epsilon$ is compact and convex. (Straightforward.)
    \item For every $\sigma_N \in \Sigma_N^\epsilon$, $B^\epsilon \left(\sigma_N \right)$ is nonempty and convex. (Straightforward.)
    \item $B^\epsilon$ has a closed graph.
\end{enumerate}

    We want to show that if $\left\{\sigma_N^{k}\right\}_{k \in \mathbb{N}}\subseteq \Sigma_N^\epsilon$ is such that $\sigma_N^{k}\rightarrow \sigma_N\in\Sigma^\epsilon_N$, and there is an image sequence $\hat{\sigma}_N^{k}$ where $\hat{\sigma}_N^{k}\in B^\epsilon\left(\sigma_N^{k}\right)$, and $\hat{\sigma}_N^{k}\rightarrow\hat{\sigma}_N$ for some $\hat{\sigma}_N\in\Sigma^\epsilon_N$, then $\hat{\sigma}_N\in B^{\epsilon}\left(\sigma_N\right)$.

    Fix a player $n$ and an information set $I\in \mathcal{I}_n$. Let $\overline{\sigma}^{I,k}$ be a conjecture that is cursed-plausible with $\sigma_N^{k}$ and similarly let $\overline{\sigma}^I$ be cursed-plausible with $\sigma^k$. Observe that $\overline{\sigma}_{-n}^{I,k} \rightarrow \overline{\sigma}_{-n}^I$ and $\rho(\overline{\sigma}^{I,k},I) \rightarrow \rho(\overline{\sigma}^{I},I)$.

    Let us define
    \begin{equation}
    \zeta^k_n \equiv \left\{\tilde{\sigma}_n \in \Sigma^\epsilon_n: \tilde{\sigma}_n(I) = \hat{\sigma}_n^{k}(I)\right\},
    \end{equation}
    \begin{equation}
    \zeta_n \equiv \left\{\tilde{\sigma}_n \in \Sigma^\epsilon_n: \tilde{\sigma}_n(I) = \hat{\sigma}_n(I)\right\}.
    \end{equation}
    Since $\hat{\sigma}_N^{k}\in B^\epsilon(\sigma_N^{k})$, it follows that for all $\tilde{\sigma}_n \in \Sigma_n^\epsilon$, we have
    \begin{equation}\label{eq:seq_LBR}
    \max_{\breve{\sigma}_n \in \zeta^k_n} \left\{ \sum_{h \in I}  \rho(h, \overline{\sigma}^{I,k},I) u_n(h, \breve{\sigma}_n, \overline{\sigma}_{-n}^{I,k}) \right\} \geq \sum_{h \in I} \rho(h, \overline{\sigma}^{I,k},I) u_n(h, \tilde{\sigma}_n, \overline{\sigma}_{-n}^{I,k}).
    \end{equation}
    Taking limits on both sides of \eqref{eq:seq_LBR} (the left side by Berge's maximum theorem), it follows that
    \begin{equation}
    \max_{\breve{\sigma}_n \in \zeta_n} \left\{ \sum_{h \in I}  \rho(h, \overline{\sigma}^{I},I) u_n(h, \breve{\sigma}_n, \overline{\sigma}_{-n}^{I}) \right\} \geq \sum_{h \in I} \rho(h, \overline{\sigma}^{I},I) u_n(h, \tilde{\sigma}_n, \overline{\sigma}_{-n}^{I}).
    \end{equation}
    Thus, we have $\hat{\sigma}_N \in B^\epsilon(\sigma_N)$, as desired.

By Kakutani's theorem, it follows that $B^\epsilon$ has a fixed point.

\subsubsection{Step III}

We now use a sequence of fixed points to prove existence of SCE. Let $\left(\sigma^k_N\right)_{k \in \mathbb{N}}$ be a sequence of fully mixed strategy profiles such that $\sigma^k_N \in B^{\epsilon^k}\left(\sigma^k_N\right)$, for $\epsilon^k \rightarrow 0$.  By definition of $B^\epsilon$, for each $k$, there exists a system of conjectures $(\overline{\sigma}^{I,k})_{I \in \mathcal{I}}$ such that the assessment $(\sigma^k_N, (\overline{\sigma}^{I,k})_{I \in \mathcal{I}})$ is cursed-plausible and moreover, for each player $n$ and each information set $I \in \mathcal{I}_n$, the self-conjecture $\overline{\sigma}_n^{I,k}$ satisfies $\sigma^k_n(I)=\overline{\sigma}_n^{I,k}(I)$ and attains the maximum in expression \eqref{eq:cursed_BR}, that is,
\begin{equation}\label{eq:self_conjecture_attains_max}
     \sum_{h \in I}  \rho(h, \overline{\sigma}^{I,k},I) u_n(h, \overline{\sigma}_n^{I,k}, \overline{\sigma}_{-n}^{I,k}) =  \max_{\hat{\sigma}_n \in \Sigma_n^{\epsilon^k}} \left\{ \sum_{h \in I}  \rho(h, \overline{\sigma}^{I,k},I) u_n(h, \hat{\sigma}_n, \overline{\sigma}_{-n}^{I,k}) \right\}.
\end{equation}

Let $(\sigma^k_N, (\overline{\sigma}^{I,k})_{I \in \mathcal{I}})$ be a sequence of such assessments. Let us define $\sigma^k \equiv (\sigma^k_N,\sigma_\nature)$. By compactness and $G$ finite, we can pass to a subsequence such that:
\begin{enumerate}
    \item The sequence of strategy profiles $\sigma^k_N$ converges.
    \item For all $I \in \mathcal{I}_N$, the sequence of conjectures $\overline{\sigma}^{I,k}$ converges.
    \item For all $I \in \mathcal{I}_N$, and all $h \in I$, the sequence of probabilities $\measure[\sigma^k](\{h\} \mid I)$ converges.
\end{enumerate}
Without loss of generality, let us suppose those properties hold for the original sequence, denoting the limits $\sigma_N \equiv \lim_{k \rightarrow \infty} \sigma_N^k$ and $\overline{\sigma}^{I} \equiv \lim_{k \rightarrow \infty} \overline{\sigma}^{I,k}$.

Next, we prove that $(\overline{\sigma}^{I})_{I \in \mathcal{I}_N}$ is a system of conjectures, that is, $\measure[\overline{\sigma}^I](I) > 0$ for all $I \in \mathcal{I}_N$. Take any information set $I \in \mathcal{I}_N$. If $I$ contains the root, then trivially $\measure[\overline{\sigma}^{I}](I) = 1 > 0$. Suppose otherwise. By convergence of $\measure[\sigma^k](\{h\} \mid I)$ for all $h \in I$ and finiteness of $I$, we can pick $h \in I$ such that $\lim_{k \rightarrow \infty} \measure[\sigma^k](\{h\} \mid I) > 0$. To reduce clutter, we suppress the braces $\{\}$ when measuring sets containing a single history.

Let us take the path of histories from the root $h_1$ to $h$; that is $(h_l)_{l = 1}^L$ with actions $(a_l)_{l = 1}^L$, such that $(h_l,a_l) = h_{l+1}$ and $(h_L,a_L) = h$. Let $(I_l)_{l=1}^L$ denote the corresponding information sets. Observe that
\begin{equation}\label{eq:chain_of_mixed}
    \measure[\overline{\sigma}^{I,k}](I  ) \geq \measure[\overline{\sigma}^{I,k}](h) = \prod_{l=1}^L \overline{\sigma}_{P(I_l)}^{I,k}(a_l, I_l).
\end{equation}

By construction of $\overline{\sigma}^{I,k}$ and then by $h_l \in I_l \subseteq F(I_l)$, we have 
\begin{equation}\label{eq:bound_identity}
    \overline{\sigma}_{P(I_l)}^{I,k}(a_l, I_l) \geq \measure[\sigma^k](\{(h,a_l): h \in F(I_l)\} \mid I \cap F(I_l)) \geq \measure[\sigma^k]((h_l,a_l) \mid I).
\end{equation}

Since $(h_l,a_l) \preceq h$, we have 
\begin{equation}\label{eq:bound_epsilon}
    \measure[\sigma^k]((h_l,a_l) \mid I) \geq \measure[\sigma^k](h \mid I).
\end{equation}

Combining \eqref{eq:chain_of_mixed}, \eqref{eq:bound_identity}, and \eqref{eq:bound_epsilon}, we have
\begin{equation}\label{eq:infoset_positive_prob}
    \measure[\overline{\sigma}^{I,k}](I)  \geq \prod_{l=1}^L \overline{\sigma}_{P(I_l)}^{I,k}(a_l, I_l) \geq \prod_{l=1}^L \measure[\sigma^k]((h_l,a_l) \mid I) \geq \measure[\sigma^k](h \mid I)^L.
\end{equation}
By continuity and then by \eqref{eq:infoset_positive_prob} we have 
\begin{equation}
    \measure[\overline{\sigma}^{I}](I) = \lim_{k \rightarrow \infty} \measure[\overline{\sigma}^{I,k}](I) \geq \lim_{k \rightarrow \infty} \measure[\sigma^k](h \mid I)^L > 0.
\end{equation}
This argument holds for any $I \in \mathcal{I}_N$, so $(\overline{\sigma}^{I})_{I \in \mathcal{I}_N}$ is a system of conjectures.

It follows that $(\sigma_N,(\overline{\sigma}^I)_{I \in \mathcal{I}_N})$ is an assessment. By construction, $(\sigma_N,(\overline{\sigma}^I)_{I \in \mathcal{I}_N})$ is cursed-consistent.

All that remains is to prove that $(\sigma_N,(\overline{\sigma}^I)_{I \in \mathcal{I}_N})$ is locally rational. Consider any player $n\in N$ and any information set $I\in\mathcal{I}_n$. For any $k$, by construction we have $\sigma^k_n(I)=\overline{\sigma}_n^{I,k}(I)$. It follows that $\sigma_n(I) = \lim_{k \rightarrow \infty} \sigma^k_n(I)= \lim_{k \rightarrow \infty} \overline{\sigma}_n^{I,k}(I) = \overline{\sigma}_n^{I}(I)$. We take limits with respect to $k$ on each side of \eqref{eq:self_conjecture_attains_max} and, by Berge's maximum theorem, conclude that
\begin{equation}
    \overline{\sigma}^I_n \in \argmax_{\hat{\sigma}^I_n} \left\{ \sum_{h \in I}  \rho(h, \overline{\sigma}^I,I) u_n(h, \hat{\sigma}^I_n, \overline{\sigma}_{-n}^I) \right\},
\end{equation}
and thus that $(\sigma_N,(\overline{\sigma}^I)_{I \in \mathcal{I}_N})$ is locally rational. It follows that $(\sigma_N,(\overline{\sigma}^I)_{I \in \mathcal{I}_N})$ is a sequential cursed equilibrium.

\subsection{Proof of Theorem \ref{thm:ICE_SCE}}

\begin{lemma}\label{lem:sim_bayes_pos}
  Given a finite simultaneous Bayesian game, any strategy profile $\sigma_N$, and any $I, I' \in \mathcal{I}_N$, if $I$ and $I'$ are compatible then we have
\begin{equation}\label{eq:pos_prob_compat}
    \measure[\sigma](I \cap F(I')) > 0.
\end{equation}
\end{lemma}
\begin{proof}
    Let state $\omega$ be such that $\{\omega\} \preceq I$. We have
    \begin{equation}
        0 < \sigma_\lambda(\omega) \leq \measure[\sigma](I) = \measure[\sigma](I \cap F(I')).
    \end{equation}
    The strict inequality follows since nature's strategy is fully mixed and the game is finite. The weak inequality is by definition of information sets for a simultaneous Bayesian game. The equality follows since every player is called to play once, and $F(I') = \cup \mathcal{I}_n$ for $n = P(I')$.
\end{proof}

By \Cref{lem:sim_bayes_pos}, we have that the assessment $(\sigma_N,(\overline{\sigma}^I)_{I \in \mathcal{I}_N})$ is cursed-plausible if and only if it is cursed-consistent, which proves the equivalence of Clause 2 and Clause 3 of \Cref{thm:ICE_SCE}. 

Next we prove the equivalence of WPCE and ICE. By \Cref{lem:sim_bayes_pos}, for any $\sigma_N$ there exists a unique system of conjectures $(\overline{\sigma}^I)_{I \in \mathcal{I}_N}$ for which:
\begin{enumerate}
    \item $(\overline{\sigma}^I)_{I \in \mathcal{I}_N}$ is cursed-plausible with $\sigma_N$.
    \item For every $n\in N$ and $I\in\mathcal{I}_n$, $\overline{\sigma}_n^I(I)=\sigma_n(I)$.
\end{enumerate}
Cursed plausibility pins down each player's conjectures about other players and nature, and self-conjectures are set to equal the strategy. Note that conjectures about future play are irrelevant in a simultaneous game. Thus, the assessment $(\sigma_N, (\overline{\sigma}^I)_{I \in \mathcal{I}_N})$ is cursed-plausible. Then, Clause $2$ of \Cref{thm:ICE_SCE} holds if and only if $\sigma_N$ is a local best-response to $(\overline{\sigma}^I)_{I \in \mathcal{I}_N}$, making the assessment $(\sigma_N, (\overline{\sigma}^I)_{I \in \mathcal{I}_N})$ locally rational.

Next we prove that $\sigma_N$ is a local best-response to $(\overline{\sigma}^I)_{I \in \mathcal{I}_N}$ if and only if $\sigma_N$ is an ICE. Take any $I \in \mathcal{I}_n$, with associated type $T_n \in \mathcal{T}_n$. Take any compatible $I' \in \mathcal{I}_m$ for $m \in N \setminus \{n\}$, and any $a \in A(I')$. By definition of a simultaneous Bayesian game, we have $F(I') = \cup \mathcal{I}_m$, and $\measure[\sigma](F(I') \mid I) = 1$. Moreover, the distribution of $m$'s actions conditional on $I$ is equal to the distribution of $m$'s actions conditional on type $T_n$. It follows that 
\begin{equation}\label{eq:SCE_equiv_players}
     \overline{\sigma}_m^I(a,I') = \measure[\sigma](\{(h,a): h \in F(I')\} \mid I \cap F(I')) = \frac{\sum_{\omega \in T_n} \sigma_\nature(\omega)\sigma_m(a,\omega)}{\sum_{\omega \in T_n} \sigma_\nature(\omega)} = \gamma_m(a_{m} \mid \sigma_m, T_n).
\end{equation}
Moreover, chance moves at a single information set containing the root,  $I_{\nature} = \{h_0\}$, so
\begin{equation}\label{eq:SCE_equiv_chance}
    \overline{\sigma}_\nature^I(\omega,I_{\nature}) = \frac{\sigma_\nature(\omega)}{\sum_{\omega' \in T_n} \sigma_\nature(\omega')}.
\end{equation}
From \eqref{eq:SCE_equiv_players} and \eqref{eq:SCE_equiv_chance}, it follows that $\sigma_N$ is a local-best response to $(\overline{\sigma}^I)_{I \in \mathcal{I}_N}$ if and only if $\sigma_N$ is an ICE. This completes the proof of \Cref{thm:ICE_SCE}.

\subsection{Proof of Theorem \ref{thm:first_price}}

A necessary condition for cursed equilibrium is that $b^{\mathrm{1P}}(x)$ must be in the set \begin{equation}\label{eq:CE_1P_obj}
    \argmax_\beta \int_{\underline{x}}^{{b}^{-1}(\beta)}({v}(x) - \beta) f_{Y_1}(\alpha) d\alpha .
    \end{equation}
    \Cref{thm:first_price} follows by rearranging the first-order condition of \eqref{eq:CE_1P_obj} and then using the same arguments as for Theorem 14 of \cite{milgrom1982theory}.

\subsection{Proof of Theorem \ref{thm:dutch}}

    The first-order condition of \eqref{eq:CCE_Dutch_obj} yields
    \begin{equation}\label{eq:dutch_FOC}
        \frac{\overline{v}(x,x)-b(x)f_{Y_1}(x \mid x)}{b'(x)} - F_{Y_1}(x \mid x) = 0.
    \end{equation}
    The ODE follows by rearranging \eqref{eq:dutch_FOC} and then using the same arguments as for Theorem 14 of \cite{milgrom1982theory}.
    
    Next we prove that the bidder is willing to wait at all prices that exceed ${b}(x)$. Take any $x' > x$. By $\overline{v}$ increasing in its first argument, we have
    \begin{equation}\label{eq:wait1}
        \overline{v}(x,x') < \overline{v}(x',x').
    \end{equation}
    By Lemma 1 of \cite{milgrom1982theory}, we have
    \begin{equation}\label{eq:wait2}
        \frac{f_{Y_1}(x' \mid x)}{F_{Y_1}(x' \mid x)} \leq \frac{f_{Y_1}(x' \mid x')}{F_{Y_1}(x' \mid x')}.
    \end{equation}
    Combining \eqref{eq:dutch_ODE}, \eqref{eq:wait1}, and \eqref{eq:wait2} yields
    \begin{equation}\label{eq:dutch_deviate}
        \frac{d}{d x'}b^{\mathrm{Dutch}}(x') > [\overline{v}(x,x') - b^{\mathrm{Dutch}}(x')]\frac{f_{Y_1}(x' \mid x)}{F_{Y_1}(x' \mid x)},
    \end{equation}
    which implies that \eqref{eq:CCE_Dutch_wait} is maximized at $\beta < b(x')$.

\subsection{Proof of Theorem \ref{thm:dutch_vs_1P}}

Theorems 2 and 4 of \cite{milgrom1982theory} yield the following lemma.
\begin{lemma}\label{lem:pair_affil}
    The random variables $(Y_1,X_1)$ are affiliated.
\end{lemma}
    We have, for all $x$,
    \begin{equation}\label{eq:sign_inner}
        \overline{v}(x,x) = \E\left[ V \mid X_1 = x, Y_1 \leq x \right] \leq \E\left[ V \mid X_1 = x, Y_1 \leq \overline{x} \right]  = r(x).
    \end{equation}
    The inequality in \eqref{eq:sign_inner} follows by \Cref{lem:pair_affil}, $w$ non-decreasing, and Theorem 5 of \cite{milgrom1982theory}. By \eqref{eq:first_price_boundary}, \eqref{eq:dutch_boundary}, and \eqref{eq:sign_inner}, we have $b^{\mathrm{Dutch}}(\underline{x}) \leq b^{\mathrm{1P}}(\underline{x})$. By \eqref{eq:first_price_ODE}, and \eqref{eq:dutch_ODE}, and \eqref{eq:sign_inner}, we have that $b^{\mathrm{Dutch}}(x) > b^{\mathrm{1P}}(x)$ implies that $\frac{d}{d x}b^{\mathrm{Dutch}}(x) \leq \frac{d}{d x}b^{\mathrm{1P}}(x)$. By Lemma 2 of \cite{milgrom1982theory}, we have \Cref{thm:dutch_vs_1P}.

\subsection{Proof of Theorem \ref{thm:silent_vs_2P}}

Now we prove \Cref{thm:silent_vs_2P}. We have, for all $x$,
    \begin{equation}
        b^{\mathrm{silent}}(x) = E[V \mid X_1 = x, Y_1 \geq x] \geq E[V \mid X_1 = x, Y_1 \geq \underline{x}] = b^{\mathrm{2P}}(x). 
    \end{equation}
    The inequality follows by \Cref{lem:pair_affil}, $w$ non-decreasing, and Theorem 5 of \cite{milgrom1982theory}

\newpage

\section{Illustrating the difference between SCE and CSE}\label{app:SCE_CSE_diff}

In this section, we contrast SCE with the solution concept proposed by \cite{FPP2023}. Their $\chi$-CSE is defined for multi-stage games with observed actions. This restriction rules out the possibility that one player is privately informed about another player's actions. In multi-stage games with observed actions, each player $n$ observes their own type $\theta_n \in \Theta_n$, and then play proceeds in stages $t = 1, \ldots, T$. At each stage $t$, players observe all actions from previous stages $(a_{n,t'})_{n \in N, t' < t}$, and then choose actions (potentially simultaneously). The set of feasible actions can depend on the public history, but not on the players' private types.



\cite{FPP2023} represent player $n$'s beliefs as a function from public histories $\mathcal{H}$ and $n$'s types $\Theta_n$ to distributions over opponent types, $\phi_n: \mathcal{H} \times \Theta_n \rightarrow \Delta(\Theta_{-n})$. A \textbf{belief system} $\phi \equiv (\phi_n)_{n \in N}$ specifies beliefs for each player.  An \textbf{assessment} $(\phi,\sigma_N)$ consists of a belief system $\phi \equiv (\phi_n)_{n \in N}$ and a behavioral strategy profile $\sigma_N$. Given an assessment $(\phi,\sigma_N)$, and a player $n$, let us define the \textbf{average opponent behavioral strategy profile} as
\begin{equation}
    \overline{\sigma}_{-n}\left(a_{-n}^t \mid h^{t-1}, \theta_n\right) \equiv \sum_{\theta_{-n}} \phi_n\left(\theta_{-n} \mid h^{t-1}, \theta_n\right) \sigma_{-n}\left(a_{-n}^t \mid h^{t-1}, \theta_{-n} \right).
\end{equation}
The key idea is that, given an assessment $(\phi,\sigma_N)$, each player responds as if at each stage, with probability $(1-\chi)$ all opponents play according to the true strategy profile $\sigma_{-n}$, and with probability $\chi$ all opponents play according to $\overline{\sigma}_{-n}$ (which is independent of their types). Let the corresponding probability of opponent actions $a_{-n}^t$ be denoted ${\sigma}_{-n}^\chi\left(a_{-n}^t \mid h^{t-1}, \theta_n, \theta_{-n}\right)$. By assumption, the randomization between $\sigma_{-n}$ and $\overline{\sigma}_{-n}$ is independent across stages. 

Let $\Psi^\chi$ be the set of assessments $(\phi,\sigma_N)$ such that $\sigma_N$ is fully mixed and each player's beliefs $\phi_n$ are derived using Bayes' rule under the assumption that opponents play according to $\sigma^{\chi}_{-n}$. An assessment $(\phi,\sigma_N)$ is a \textbf{$\chi$-cursed sequential equilibrium} if:
\begin{enumerate}
    \item $(\phi,\sigma_N)$ is in the closure of $\Psi^\chi$, and
    \item for each public history $h^t$, each player $n$, and each type $\theta_n$, the action distribution $\sigma_n(h^t,\theta_n)$ maximizes type $\theta_n$'s conditional expected payoff given belief $\phi_n$, when play continues according to $(\sigma_n,\sigma_{-n}^\chi)$.
\end{enumerate}

We illustrate the difference between SCE and $\chi$-CSE with examples. Consider \Cref{fig:diff_1}. This is a two-stage game with observed actions, with two types for player $2$, denoted $\theta_2^L$ and $\theta_2^R$. Player $1$ moves in stage $1$, and then player $2$ observes $1$'s action and moves in stage $2$. It is a strictly dominant strategy for player $2$ to play $l$ after $L$ and $r$ after $R$. It is a Bayesian Nash equilibrium for player $1$ to respond by playing $b$, and this equilibrium is unique.  However, under SCE, player $1$ neglects how $2$'s actions depend on $2$'s private information; in the unique SCE, player $2$ plays their dominant strategy and player $1$ plays $a$.  This strategy profile is also a $\chi$-CSE, for any $\chi \geq 2/3$.

\begin{figure}[h]
\centering
\begin{subfigure}[T]{.48\textwidth}
  \centering
\begin{tikzpicture}[font=\footnotesize,edge from parent/.style={draw,thick}]
    \tikzstyle{solid}=[circle,draw,inner sep=1.2,fill=black];
    \tikzstyle{hollow}=[circle,draw,inner sep=1.2]
    \tikzstyle{solid_red}=[circle,draw=red,inner sep=1.7,fill=red];
    \tikzstyle{solid_mag}=[circle,draw=magenta,inner sep=1.7,fill=red];
    \tikzstyle{level 1}=[level distance=10mm,sibling distance=32mm]
    \tikzstyle{level 2}=[level distance=12mm,sibling distance=16mm]
    \tikzstyle{level 3}=[level distance=12mm,sibling distance=8mm]
     \node(0)[hollow]{}
     child{node[solid]{}
        child{node[solid]{}
            child{node[below]{}
            edge from parent node[left]{$l$}
            }
            child{node[below]{}
            edge from parent node[right]{$r$}
            }
        edge from parent node[above left]{$a$}
        }
        child{node[solid]{}
            child{node[below]{}
            edge from parent node[left]{$l$}
            }
            child{node[below]{}
            edge from parent node[right]{$r$}
            }
        edge from parent node[above right]{$b$}
        }
     edge from parent node[above left]{$L[.5]$}
     }
     child{node[solid]{}
        child{node[solid]{}
            child{node[below]{}
            edge from parent node[left]{$l$}
            }
            child{node[below]{}
            edge from parent node[right]{$r$}
            }
        edge from parent node[above left]{$a$}
        }
        child{node[solid]{}
            child{node[below]{}
            edge from parent node[left]{$l$}
            }
            child{node[below]{}
            edge from parent node[right]{$r$}
            }
        edge from parent node[above right]{$b$}
        }
     edge from parent node[above right]{$R[.5]$}
     }
     ;

    \node[below]at(0-1-1-1){\stackon{$1$}{$-1$}};
    \node[below]at(0-1-1-2){\stackon{$0$}{$0$}};
    \node[below]at(0-1-2-1){\stackon{$1$}{$0$}};
    \node[below]at(0-1-2-2){\stackon{$0$}{$0$}};
    \node[below]at(0-2-1-1){\stackon{$0$}{$2$}};
    \node[below]at(0-2-1-2){\stackon{$1$}{$0$}};
    \node[below]at(0-2-2-1){\stackon{$0$}{$0$}};
    \node[below]at(0-2-2-2){\stackon{$1$}{$0$}};
   \draw[dashed](0-1)to(0-2);
     \node[above,yshift=2]at(0){nature};
     \node[above,yshift=2]at(0-1){1}; 
     \node[above,yshift=2]at(0-2){1}; 
     \node[above,yshift=2]at(0-1-1){2};
     \node[above,yshift=2]at(0-1-2){2};
     \node[above,yshift=2]at(0-2-1){2};
     \node[above,yshift=2]at(0-2-2){2};
\end{tikzpicture}
  \caption{Player $1$ goes first}
  \label{fig:diff_1}
\end{subfigure}
\begin{subfigure}[T]{.48\textwidth}
  \centering
\begin{tikzpicture}[font=\footnotesize,edge from parent/.style={draw,thick}]
    \tikzstyle{solid}=[circle,draw,inner sep=1.2,fill=black];
    \tikzstyle{hollow}=[circle,draw,inner sep=1.2]
    \tikzstyle{solid_red}=[circle,draw=red,inner sep=1.7,fill=red];
    \tikzstyle{solid_mag}=[circle,draw=magenta,inner sep=1.7,fill=red];
    \tikzstyle{level 1}=[level distance=10mm,sibling distance=32mm]
    \tikzstyle{level 2}=[level distance=12mm,sibling distance=16mm]
    \tikzstyle{level 3}=[level distance=12mm,sibling distance=8mm]
     \node(0)[hollow]{}
     child{node[solid]{}
        child{node[solid]{}
            child{node[below]{}
            edge from parent node[left]{$a$}
            }
            child{node[below]{}
            edge from parent node[right]{$b$}
            }
        edge from parent node[above left]{$l$}
        }
        child{node[solid]{}
            child{node[below]{}
            edge from parent node[left]{$a$}
            }
            child{node[below]{}
            edge from parent node[right]{$b$}
            }
        edge from parent node[above right]{$r$}
        }
     edge from parent node[above left]{$L[.5]$}
     }
     child{node[solid]{}
        child{node[solid]{}
            child{node[below]{}
            edge from parent node[left]{$a$}
            }
            child{node[below]{}
            edge from parent node[right]{$b$}
            }
        edge from parent node[above left]{$l$}
        }
        child{node[solid]{}
            child{node[below]{}
            edge from parent node[left]{$a$}
            }
            child{node[below]{}
            edge from parent node[right]{$b$}
            }
        edge from parent node[above right]{$r$}
        }
     edge from parent node[above right]{$R[.5]$}
     }
     ;

    \node[below]at(0-1-1-1){\stackon{$1$}{$-1$}};
    \node[below]at(0-1-1-2){\stackon{$1$}{$0$}};
    \node[below]at(0-1-2-1){\stackon{$0$}{$0$}};
    \node[below]at(0-1-2-2){\stackon{$0$}{$0$}};
    \node[below]at(0-2-1-1){\stackon{$0$}{$2$}};
    \node[below]at(0-2-1-2){\stackon{$0$}{$0$}};
    \node[below]at(0-2-2-1){\stackon{$1$}{$0$}};
    \node[below]at(0-2-2-2){\stackon{$1$}{$0$}};
    \draw[dashed,bend left=30](0-1-1)to(0-2-1);
    \draw[dashed,bend right=30](0-1-2)to(0-2-2);
     \node[above,yshift=2]at(0){nature};
     \node[above,yshift=2]at(0-1){2}; 
     \node[above,yshift=2]at(0-2){2}; 
     \node[above,yshift=2]at(0-1-1){1};
     \node[above,yshift=2]at(0-1-2){1};
     \node[above,yshift=2]at(0-2-1){1};
     \node[above,yshift=2]at(0-2-2){1};
\end{tikzpicture}
  \caption{Player $2$ goes first}
  \label{fig:diff_2}
\end{subfigure}%
\caption{Examples that distinguish SCE and $\chi$-CSE}\label{fig:diff_join_1_2}
\end{figure}

SCE allows for cursedness about endogenous information, whereas $\chi$-CSE does not. To illustrate this, suppose that we alter the game in \Cref{fig:diff_1}, replacing nature with a fictitious player $0$, who is indifferent between all terminal histories. There is an equivalent SCE in which player $0$ plays $(L, .5; R, .5)$, player $1$ plays $a$, and player $2$ plays their dominant strategy---in this example, player $1$ is cursed about endogenous information.

By contrast, $\chi$-CSE does not allow player $2$ to make mistakes in the game with a fictitious player. (This is still a two-stage game with observed actions, so $\chi$-CSE is still well-defined; in stage $1$, players $0$ and $1$ choose simultaneously, and in stage $2$ player $2$ observes the earlier actions and then chooses.)  Now every player's type set is singleton, so $\chi$-CSE reduces to sequential equilibrium for any $\chi \in [0,1]$.  Thus, $\chi$-CSE predicts that if player $0$ plays $L$ with positive probability, then player $1$ plays $b$. This illustrates how $\chi$-CSE does not allow for cursedness about endogenous information.

SCE distinguishes between hypothetical and observed events, whereas $\chi$-CSE does not. To illustrate this, consider the game in \Cref{fig:diff_2}. This is also a two-stage game, but the stages are reversed compared to \Cref{fig:diff_1}: Player $2$ moves in stage $1$, and then player $1$ observes $2$'s action and moves in stage $2$. As before, it is strictly dominant for player $2$ to play $l$ after $L$ and $r$ after $R$. Given this behavior, at player $1$'s information set $\{Ll, Rl\}$, any cursed-plausible conjecture specifies that nature has played $L$ with probability $1$. Thus, in every SCE player $1$ plays $b$ after player $2$ plays $l$; observing player $2$'s action prevents player $1$ from making a mistake.

By contrast, $\chi$-CSE does not predict that observing player $2$'s action helps player $1$ avoid mistakes. Under $\chi$-CSE player $1$ responds as if with probability $(1-\chi)$, player $2$ plays $l$ given type $\theta_2^L$ and $r$ given type $\theta_2^R$, and with probability $\chi$, player 2 plays $(l, .5; r, .5)$ independently of their type. Thus, if assessment $(\phi,\sigma_n)$ is a $\chi$-CSE, then player $1$ believes that after the public history in which $2$ plays $l$, player $2$ has type $\theta_2^L$ with probability $1 - \chi/2$. Playing $a$ is a best-response to this belief for $\chi \geq 2/3$. Thus, $\chi$-CSE makes the same prediction regardless of whether player $1$ chooses hypothetically (as in \Cref{fig:diff_1}) or after observing $2$'s action (as in \Cref{fig:diff_2}).

SCE players can be dynamically inconsistent, because they behave differently when choosing prospectively (\Cref{fig:diff_1}) compared to choosing retrospectively (\Cref{fig:diff_2}). \Cref{ex:groucho} illustrates these inconsistencies: Under SCE, Groucho joins the club, but resigns after learning that the club accepted him.

This dynamic inconsistency does not arise under $\chi$-CSE. If $\chi$ is high enough that Groucho accepts, then $\chi$ is too high for Groucho to resign. In any $\chi$-CSE of \Cref{ex:groucho}, Groucho believes at the first information set that with probability $\chi$ the club plays $(d, \frac{2}{3}; a, \frac{1}{3})$ independently of its type, and with probability $(1-\chi)$ the club plays $a$ in state $\omega_1$ and $d$ in state $\omega_2$. It is a best-response for Groucho to play $a$ only if $\chi \geq 3/4$. But at the last information set Groucho believes that the state is $\omega_1$ with probability $\chi \frac{1}{3} + (1-\chi) 1$, a weighted average of the fully cursed belief and the Bayesian belief. At that information set, it is a best-response for Groucho to resign only if $\chi \leq 3/8$. Thus, under $\chi$-CSE (for any parameter $\chi$) if Groucho accepts at the first information set, then he confirms at the last information set.

In summary, we see the key distinctions between SCE and $\chi$-CSE as follows: First, $\chi$-CSE predicts only cursedness about types, whereas SCE also predicts cursedness about endogenous information. Second, SCE predicts that players make different inferences from hypothetical events and observed events, whereas $\chi$-CSE predicts that they do not. Third, these different inferences lead players to be dynamically inconsistent under SCE, whereas similar inconsistencies do not arise under $\chi$-CSE.

\section{WPCE for infinite action sets}\label{sec:WPCE_infinite}

WPCE extends to games with infinite action sets. For this purpose, we must now be more explicit about measure. Suppose we have some sigma-algebra $\mathcal{Z}$ on the set of terminal histories $Z$, and the strategy profile $\sigma$ induces a probability measure $\measure[\sigma]$ on $\mathcal{Z}$.\footnote{We require that each utility function is measurable with respect to that sigma-algebra, and that the sigma-algebra includes the sets $\{z \in Z: I \prec \{z\}\}$ for each $I \in \mathcal{I}$ and the sets $\{z \in Z: F \prec \{z\}\}$ for each $F \in \mathcal{F}$.}  As before, we abuse notation to treat each set of histories $Q$ as equivalent to the set of terminal histories that succeed $Q$. Take two information sets $I$ and $I'$, such that $I \cap F(I') \in \mathcal{Z}$ and the conditional distribution over $I \cap F(I')$ exists, which we denote $\measure[\sigma, I \cap F(I')]$. Suppose we have some sigma-algebra $\mathcal{A}$ on the set of actions $A(I')$. Consider the function $\alpha: I \cap F(I') \rightarrow A(I')$, defined by $\alpha(z) \equiv \{a \in A(I'): \exists h \in F(I'): (h,a) \preceq z$\}. If $\alpha$ is a measurable function, then we can define the image measure of $\measure[\sigma, I \cap F(I')]$ under $\alpha$; that is, the probability measure $\tilde{\measure}_{A(I')}[\sigma, I \cap F(I')]$ on $\mathcal{A}$ such that
\begin{equation}\label{eq:conditional_measure}
    \tilde{\measure}_{A(I')}[\sigma, I \cap F(I')](Q) \equiv \measure[\sigma, I \cap F(I')](\alpha^{-1}(Q)) \text{ for each $Q \in \mathcal{A}$}.
\end{equation}

The definition of WPCE offered in \Cref{sec:def_existence} extends to infinite action sets, with just two modifications.

First, instead of requiring that the conjecture $\overline{\sigma}^I$ reaches information set $I$ with positive probability, we require that $\overline{\sigma}^I$ is such that the conditional distribution over $I$ exists. 

Second, when defining cursed-plausibility, we replace Clause 2 of \Cref{def:cursed_plausible} with the following requirement: For all $m \neq n$ and all $I' \in \mathcal{I}_m$, for $\sigma \equiv (\sigma_N,\sigma_\nature)$, we have
    \begin{equation}\label{eq:WPCE_other}
        \overline{\sigma}_m^I(I') = \tilde{\measure}_{A(I')}[\sigma, I \cap F(I')]
    \end{equation}
    if the conditional distribution defined by \eqref{eq:conditional_measure} exists.

There are many obstacles to extending sequential equilibrium to games with infinite action sets.\footnote{For a discussion of these obstacles and a proposed solution, see \cite{myerson2020perfect}.} SCE is built on the sequential equilibrium machinery, so it seems unlikely that SCE has a simple extension to such games.

\section{An example with $\mathcal{F} = \mathcal{I}$}\label{sec:FI_non_equiv}

Sequential equilibrium and sequential cursed equilibrium may not coincide, even if every coarse set contains exactly one information set (that is, $\mathcal{F} = \mathcal{I}$). Even in these cases, the cursed-plausible conjectures can neglect to condition on strategically-relevant hypothetical events. For an example, consider the extensive game in \Cref{fig:IF_example}. Players $1$ and $2$ are playing matching pennies against each other. This game has a unique Nash equilibrium, in which $1$ plays $\left(H\ .5,T\ .5\right)$, $2$ plays $\left(h\ .5,t\ .5\right)$ and $3$ plays $a$.

\begin{figure}[h]
    \centering
\begin{tikzpicture}[font=\footnotesize,edge from parent/.style={draw,thick}]
    \tikzstyle{solid}=[circle,draw,inner sep=1.2,fill=black];
    \tikzstyle{hollow}=[circle,draw,inner sep=1.2]
    \tikzstyle{solid_red}=[circle,draw=red,inner sep=1.7,fill=red];
    \tikzstyle{solid_mag}=[circle,draw=magenta,inner sep=1.7,fill=red];
    \tikzstyle{level 1}=[level distance=7mm,sibling distance=40mm]
    \tikzstyle{level 2}=[level distance=14mm,sibling distance=20mm]
    \tikzstyle{level 3}=[level distance=10mm,sibling distance=10mm]
     \node(0)[hollow]{}
     child{node[solid]{}
        child{node[below]{}
        edge from parent node[above left,yshift=-.5mm]{$h{\color{blue} [.5]}$}
        }
        child{node[solid]{}
            child{node[below]{}
            edge from parent node[left]{$a$}
            }
            child{node[below]{}
            edge from parent node[right]{$b$}
            }
        edge from parent node[above right,yshift=-.5mm]{$t{\color{blue} [.5]}$}
        }
     edge from parent node[above left]{$H{\color{blue} [.5]}$}
     }
     child{node[solid]{}
        child{node[solid]{}
            child{node[below]{}
            edge from parent node[left]{$a$}
            }
            child{node[below]{}
            edge from parent node[right]{$b$}
            }
        edge from parent node[above left,yshift=-.5mm]{$h{\color{blue} [.5]}$}
        }
        child{node[solid]{}
            child{node[below]{}
            edge from parent node[left]{$a$}
            }
            child{node[below]{}
            edge from parent node[right]{$b$}
            }
        edge from parent node[above right,yshift=-.5mm]{$t{\color{blue} [.5]}$}
        }
     edge from parent node[above right]{$T{\color{blue} [.5]}$}
     }
     ;
    \draw[blue, ->, line width=1.5pt, -{Triangle[length=8pt]},] (0-1-2) -- ($(0-1-2)!.97!(0-1-2-2)$);
    \draw[blue, ->, line width=1.5pt, -{Triangle[length=8pt]},] (0-2-1) -- ($(0-2-1)!.97!(0-2-1-2)$);
    \draw[blue, ->, line width=1.5pt, -{Triangle[length=8pt]},] (0-2-2) -- ($(0-2-2)!.97!(0-2-2-2)$);
    \node[below]at(0-1-1){\stackon{\stackon{$0$}{$0$}}{$1$}};
    \node[below]at(0-1-2-1){\stackon{\stackon{$0$}{$1$}}{$0$}};
    \node[below]at(0-1-2-2){\stackon{\stackon{$0$}{$1$}}{$0$}};
    \node[below]at(0-2-1-1){\stackon{\stackon{$12$}{$1$}}{$0$}};
    \node[below]at(0-2-1-2){\stackon{\stackon{$0$}{$1$}}{$0$}};
    \node[below]at(0-2-2-1){\stackon{\stackon{$0$}{$0$}}{$1$}};
    \node[below]at(0-2-2-2){\stackon{\stackon{$7$}{$0$}}{$1$}};
    \draw[dashed](0-1)to(0-2);
    \draw[dashed](0-1-2)to(0-2-2);
     \node[above,yshift=2]at(0){1};
     \node[above,yshift=2]at(0-1){2}; 
     \node[above,yshift=2]at(0-2){2}; 
     \node[above,yshift=2]at(0-1-2){3};
     \node[above,yshift=2]at(0-2-1){3};
     \node[above,yshift=2]at(0-2-2){3};
\end{tikzpicture}
    \caption{}
    \label{fig:IF_example}
\end{figure}

The strategy profile in \Cref{fig:IF_example} is not a Nash equilibrium, and hence not a sequential equilibrium; player $3$ can profitably deviate to play $a$. However, this strategy profile is a SCE. Conditional on player $3$'s information set $I_3$ and player $1$'s coarse set, $1$ plays $T$ with probability $\frac{2}{3}$, so the cursed-plausible conjecture at $I_3$ specifies that $1$ plays $\left(H\ \frac{1}{3},T\ \frac{2}{3}\right)$. Conditional on $I_3$ and player $2$'s coarse set, player $2$ plays $t$ with probability $\frac{2}{3}$, so the cursed-plausible conjecture at $I_3$ specifies that $2$ plays $\left(h\ \frac{1}{3},t\ \frac{2}{3}\right)$. Facing this conjecture, it is a local best response for $3$ to play $b$.

In this example, even though player $1$ and player $2$'s behavior are uncorrelated \textit{ex ante}, they are correlated conditional on $I_3$. At $I_3$, player $3$ is told ``either $1$ played $T$ or $2$ played $t$". Player $3$ correctly infers that player $2$ played $t$ with probability $\frac{2}{3}$. But what matters is that conditional on player $1$ playing $T$, player $2$ played $t$ with probability $\frac{1}{2}$, and player $3$ fails to condition on this hypothetical event. If we altered this example so that player $3$ observes that $1$ played $T$ (deleting the leftmost history in $I_3$), then player $3$ would respond optimally, and SCE and sequential equilibrium would predict the same behavior.

\section{Clock cursed equilibrium for the canonical English auction}\label{sec:CCE_canon_def}

A \textbf{clock cursed equilibrium of the canonical English auction} is a profile of bidding functions $({b}_k)_{k \in \{0,\ldots, |N|-2\}}$ such that for all $(y_{|N|-1},\ldots,y_{|N|-k})$ and for all $x > x' \geq y_{|N|-k}$, there exists $\beta > b_k(x' \mid y_{|N|-1},\ldots,y_{|N|-k})$ that maximizes
\begin{equation}\label{eq:CCE_canon_wait}
    \E \left[ (v_k(x,y_{|N|-1},\ldots,y_{|N|-k},x') - B)1_{B \leq \beta} \mid \xi, Y_{|N|-k-1} \geq x' \right],
\end{equation}
and for all $x \geq y_{|N|-k}$, the bid $b_k(x \mid y_{|N|-1},\ldots,y_{|N|-k})$ is in
\begin{equation}\label{eq:CCE_canon_obj}
    \argmax_{\beta} \E \left[ (v_k(x,y_{|N|-1},\ldots,y_{|N|-k},x) - B)1_{B \leq \beta} \mid \xi, Y_{|N|-k-1} \geq x \right],
\end{equation}
for the random variable $B \equiv b_{|N|-2}\left(Y_1 \mid Y_{|N|-1}, \ldots , Y_2\right)$, where $\xi$ denotes the event $X_1 = x, Y_{|N|-1} = y_{|N|-1},\ldots, Y_{|N|-k} = y_{|N|-k}$.

\section{Causal Sequential Cursed Equilibrium}\label{app:causal_SCE}

We define causal sequential cursed equilibrium. Essentially, at each information set, we keep track of one conjecture for each available action, enabling the active player to understand how choosing different actions affects the other players' behavior.

Given information set $I \in \mathcal{I}_n$, let $\sigma_n / (I,a)$ denote the strategy that is identical to $\sigma_n$ except that it plays $a$ for certain at $I$. Let $\mathcal{Q} \equiv \{(I,a): I \in \mathcal{I}_N, a \in A(I)\}$. Given $(I,a) \in \mathcal{Q}$, an \textbf{action-conjecture} $\sigma^{I,a}$ specifies play at all $I'$ such that $\{({h},{a}): {h} \in I\} \preceq I'$ or \textit{vice versa}. 

\begin{definition}\label{def:causal_cursed_plausible}
For player $n \in N$, information set $I \in \mathcal{I}_n$, and action $a \in A(I)$, the action-conjecture $\overline{\sigma}^{I,a}$ is \textbf{causal cursed-plausible} with strategy profile $\sigma_N$ if
\begin{enumerate}
    \item For every $I' \in \mathcal{I}_n$ that precedes $I$, the distribution $\overline{\sigma}^{I,a}_n(I')$ places probability $1$ on the action that does not rule out reaching $I$.
    \item For every player $m \neq n$, every information set $I' \in \mathcal{I}_m$, and every action $a' \in A(I')$, for $\sigma \equiv (\sigma_N,\sigma_\nature)$, we have
    \begin{equation}
        \overline{\sigma}_m^{I,a}(a',I') = \measure[\left(\sigma_{n}/(I,a),\sigma_{-n}\right)](\{(h,a'): h \in F(I')\} \mid I \cap F(I'))
    \end{equation}
    if $\measure[\sigma](I \cap F(I')) > 0$.
    \item For all $m \neq n$, the partial strategy  $\overline{\sigma}_m^{I,a}$ is coarse.
    \item $\overline{\sigma}_n^{I,a}$ plays action $a$ for certain at information set $I$.
\end{enumerate}
\end{definition}

\Cref{def:causal_cursed_plausible} modifies \Cref{def:cursed_plausible} in two ways. First, Clause 2 of \Cref{def:causal_cursed_plausible} uses the probability measure induced by $\left(\sigma_{n}/(I,a),\sigma_{-n}\right)$ instead of that induced by $\sigma$. Second, Clause 4 of \Cref{def:causal_cursed_plausible} requires the action-conjecture to place probability $1$ on the associated action. 

A \textbf{causal system of conjectures} is a collection $(\sigma^{I,a})_{(I,a) \in \mathcal{Q}}$, and a \textbf{causal assessment} is a tuple $(\sigma_N, (\overline{\sigma}^{I,a})_{(I,a) \in \mathcal{Q}})$ consisting of a strategy profile and a causal system of conjectures. A causal assessment is causal cursed-plausible if the requirements of \Cref{def:causal_cursed_plausible} hold for every $n \in N$, $I \in \mathcal{I}_n$, and $a \in A(I)$. A causal assessment is \textbf{causal cursed-consistent} if it is the limit of a sequence of causal cursed-plausible assessments, each with a fully mixed strategy profile.

A causal assessment $(\sigma_N, (\overline{\sigma}^{I,a})_{(I,a) \in \mathcal{Q}})$ is \textbf{locally rational} if for all $n \in N$, all $I \in \mathcal{I}_n$, and all $a \in A(I)$, if $\sigma_n(I,a) > 0$ then
\begin{equation}
    \sum_{h \in I}  \rho\left(h, \overline{\sigma}^{I,{a}},I\right) u_n\left(h, \overline{\sigma}^{I,a}\right)  = \max_{\hat{a} \in A(I)}\max_{\hat{\sigma}_n} \left\{ \sum_{h \in I}  \rho\left(h, \overline{\sigma}^{I,\hat{a}},I\right) u_n\left(h, \hat{\sigma}_n / (I,\hat{a}), \overline{\sigma}_{-n}^{I,\hat{a}}\right) \right\}.
\end{equation}

A \textbf{causal sequential cursed equilibrium} is a causal assessment that is causal cursed-consistent and locally rational. We omit the existence proof, since it is a straightforward adaptation of the proof of \Cref{thm:SCE_existence}.

\end{document}